\documentclass[journal]{IEEEtran}%
\pdfoutput=1
\usepackage{amsfonts}
\usepackage{amsmath}
\usepackage{amsthm}
\usepackage{amssymb}
\usepackage{graphicx}
\usepackage{algorithm}
\usepackage{algorithmic}%
\providecommand{\U}[1]{\protect\rule{.1in}{.1in}}
\newtheorem{theorem}{Theorem}
\newtheorem{definition}[theorem]{Definition}
\begin{document}

\title{Minimal-memory, non-catastrophic, polynomial-depth quantum convolutional encoders}
\author{Monireh Houshmand, Saied Hosseini-Khayat, and Mark M. Wilde \thanks{Monireh
Houshmand and Saied Hosseini-Khayat are with the Department of Electrical
Engineering, Ferdowsi University of Mashhad, Mashhad, Iran. Mark M. Wilde is
with the School of Computer Science, McGill University, Montreal, Qu\'{e}bec,
Canada H3A 2A7. (E-mails: monireh.houshmand@gmail.com; shk@ieee.org; mark.wilde@mcgill.ca)}}
\maketitle

\begin{abstract}
Quantum convolutional coding is a technique for encoding a stream of quantum
information before transmitting it over a noisy quantum channel.
Two important goals in the design of quantum convolutional encoders are to
minimize the memory required by them and to avoid the catastrophic propagation
of errors. In a previous paper, we determined minimal-memory,
non-catastrophic, polynomial-depth encoders for a few exemplary quantum
convolutional codes. In this paper, we elucidate a general technique for
finding an encoder of an arbitrary quantum convolutional code such that the
encoder possesses these desirable properties. We also provide an elementary
proof that these encoders are non-recursive. Finally, we apply our
technique to many quantum convolutional codes from the literature.

\end{abstract}

\section{Introduction}

A quantum convolutional code is a particular type of quantum error-correcting
code~\cite{PhysRevA.52.R2493,thesis97gottesman,PhysRevLett.77.793}\ that is
well-suited for the regime of quantum
communication~\cite{PhysRevLett.91.177902,ieee2007forney,GR06}. In this
regime, we assume that a sender and receiver have free access to local,
noiseless quantum computers, and the only source of noise is due to a quantum
communication channel connecting the sender to the receiver. The advantage of
the convolutional approach to quantum error correction is that the repeated
application of the same unitary operation encodes a stream of quantum
information, and the complexity of the decoding algorithm is linear in the
length of the qubit stream~\cite{PTO09}. Many researchers have generated a
notable literature on this topic, addressing various issues such as code
constructions~\cite{cwit2007aly,arx2007aly}, encoders and
decoders~\cite{GR06,GR06b,GR07}, and alternate paradigms with entanglement
assistance~\cite{WB07,WB08,WB09}\ or with gauge qubits and classical
bits~\cite{WB08a}. Perhaps more importantly for the quantum communication
paradigm, quantum convolutional codes are the constituents of a quantum serial
turbo code~\cite{PTO09}, and these codes are among the highest performing
codes in both the standard~\cite{PTO09} and entanglement-assisted
settings~\cite{WH10b}.

One of the most important parameters for a quantum convolutional encoder is
the size of its memory, defined as the number of qubits that are fed from its
output into the next round of encoding. A quantum convolutional encoder with a
large memory is generally more difficult to implement because it requires
coherent control of a large number of qubits. Furthermore, the complexity of
the decoding algorithm for a quantum convolutional code is linear in the
length of the qubit stream, but it is exponential in the size of the
memory~\cite{PTO09}. The decoding algorithm will thus have more delay for a
larger memory, and this could potentially lead to further errors in the more
practical setting where there is local noise at the receiving end. Therefore,
an interesting and legitimate question is to determine the minimal number of
memory qubits required to implement a given quantum convolutional code.

Another property that any good quantum convolutional \textit{decoder} should
possess is non-catastrophicity. As the name suggests, the consequences of
decoding with a catastrophic decoder are disastrous---it can propagate some
uncorrected errors infinitely throughout the decoded information qubit stream
and the receiver will not know that this is happening. We should clarify that
catastrophicity is a property of the decoder because the only errors that
occur in the communication paradigm are those due to the channel, and thus the
decoder (and not the encoder) has the potential to propagate uncorrected
errors. Though, we could say just as well that catastrophicity is a property
of an encoder if the decoder is the exact inverse of the encoder (as is the
case in Ref.~\cite{PTO09}). Either way, since the property of
non-catastrophicity is essential and having a minimal memory is highly
desirable, we should demand for our encoders and decoders to be both
minimal-memory and non-catastrophic.

The minimal-memory/non-catastrophic question is essentially understood for the
case of irreversible encoders for classical convolutional
codes~\cite{F70,book1999conv} by making use of ideas in linear system theory.
Though, these results at the surface do not appear to address the case of
reversible classical encoders, which would be more relevant for answering the
minimal-memory/non-catastrophic question in the quantum case.

In Refs.~\cite{GR06,GR06b,GR07}, Grassl and R\"{o}tteler proposed an algorithm
to construct non-catastrophic quantum circuits for encoding quantum
convolutional codes. Their encoders there do not have a convolutional
structure, and their work did not address how much quantum memory their
encoders would require for implementation. In follow-up work, we found a
minimal-memory realization of a Grassl-R\"{o}tteler encoder by performing a
longest path search through a \textquotedblleft commutativity
graph\textquotedblright\ that corresponds to the encoder~\cite{HHW10,HH11}.
Our approach was generally sub-optimal because there exist many encoders for a
given convolutional code---starting from a Grassl-R\"{o}tteler encoder and
finding the minimal-memory representation for it does not necessarily lead to
a minimal-memory encoder among all possible representations of the code. Also,
the complexity of the Grassl-R\"{o}tteler algorithm for computing an encoding circuit could
be exponential in general, resulting in an encoding circuit with exponential
depth~\cite{GR06,GR06b}.

The purpose of the present paper is to elucidate the technique of
Ref.~\cite{WHK11} in full detail. The encoders resulting from our technique
are convolutional and possess the aforementioned desirable properties
simultaneously---they are minimal-memory, non-catastrophic, and have an
$O(n^{2})$ depth, where $n$ is the frame size of code. In addition, we prove
that the resulting encoders are non-recursive. Ref.~\cite{PTO09} already proved
that all non-catastrophic encoders are non-recursive, but our proof of this fact for the encoders
studied here is arguably much simpler than the proof of Theorem~1 in Ref.~\cite{PTO09}.
Interestingly, the
essence of our technique for determining an encoder is commutation relations, which often are lurking
behind many fundamental questions in quantum information theory. The
commutation relations that are relevant for our technique are those for the
Pauli operators acting on the memory qubits. An upshot of our technique for
minimizing memory is that it is similar to one in
Refs.~\cite{arx2008wildeOEA,PhysRevA.79.062322}\ for finding the minimal
number of entangled bits required in an entanglement-assisted quantum
error-correcting code~\cite{BDH06}. This result is perhaps unsurprising in
hindsight, given that an encoder generally entangles information qubits and
ancilla qubits with the memory qubits before sending encoded qubits out over
the channel.

This paper is organized as follows. For the sake of completeness, we begin by
reviewing the definition of a quantum convolutional code. We then review our
technique from Ref.~\cite{WHK11}\ for determining a quantum convolutional
encoder for a given set of stabilizer generators, and we prove a theorem
concerning the consistency of these generators with commutation relations of
the encoder. Section~\ref{sec:mem-comm} introduces the idea of a memory
commutativity matrix that is rooted in ideas from Ref.~\cite{WHK11}.
Section~\ref{sec:state-diagram}\ reviews the state diagram for a quantum
convolutional encoder~\cite{PTO09,V71,book1999conv,McE02}, and the section
after it reviews catastrophicity. All of the above sections feature a
\textquotedblleft running example\textquotedblright\ that is helpful in
illustrating the main concepts. Section~\ref{sec:main-theorem} details our
main results, which are sufficient conditions for any quantum convolutional
encoder to be both minimal-memory and non-catastrophic. These sufficient
conditions apply to the memory commutativity matrix of the quantum
convolutional encoder. Section~\ref{sec:non-rec} then proves that the
encoders studied in Section~\ref{sec:main-theorem} are non-recursive. Finally, we conclude in Section~\ref{sec:conclusion}%
\ with a summary and a list of open questions, and
the appendix gives many examples of quantum convolutional
codes from Refs.~\cite{ieee2007forney,GR07} for which we can find
minimal-memory, non-catastrophic encoders.

\section{Quantum Convolutional Codes}

In this section, we recall some standard facts and then review the definition
of a quantum convolutional code.
A Pauli sequence is a countably-infinite tensor product of Pauli matrices:%
\[
\mathbf{A}=\bigotimes\limits_{i=0}^{\infty}A_{i},
\]
where each operator $A_{i}$ in the sequence is an element of the Pauli
group~$\Pi\equiv\left\{  I,X,Y,Z\right\}  $. Let $\Pi^{\mathbb{Z}^{+}}$ denote
the set of all Pauli sequences. A Pauli sequence is finite-weight if only
finitely many operators~$A_{i}$ in the sequence are equal to $X$, $Y$, or $Z$,
and it is an infinite-weight sequence otherwise.

\begin{definition}
[Quantum Convolutional Code]\label{def:QCC}A rate-$k/n$ quantum convolutional
code admits a representation with a basic set $\mathcal{G}_{0}$ of $n-k$
generators and all of their $n$-qubit shifts:%
\[
\mathcal{G}_{0}\equiv\left\{  \mathbf{G}_{i}\in\Pi^{\mathbb{Z}^{+}}:1\leq
i\leq n-k\right\}  .
\]
In order to form a quantum convolutional code, these generators should commute
with themselves and all of the $n$-qubit shifts of themselves and the other generators.
\end{definition}

Equivalently, a rate-$k/n$ quantum convolutional code is specified by $n-k$
generators $h_{1}$, $h_{2}$, $\ldots$, $h_{n-k}$, where%
\begin{equation}%
\begin{array}
[c]{cc}%
h_{1} & =\\
h_{2} & =\\
\vdots & \\
h_{n-k} & =
\end{array}%
\begin{array}
[c]{c}%
h_{1,1}\\
h_{2,1}\\
\vdots\\
h_{n-k,1}%
\end{array}
\begin{array}
[c]{c}%
\vert\\
\vert\\
\vdots\\
\vert%
\end{array}
\begin{array}
[c]{c}%
h_{1,2}\\
h_{2,2}\\
\vdots\\
h_{n-k,2}%
\end{array}
\begin{array}
[c]{c}%
\vert\\
\vert\\
\vdots\\
\vert%
\end{array}
\begin{array}
[c]{c}%
\cdots\\
\cdots\\
\ \\
\cdots
\end{array}
\begin{array}
[c]{c}%
\vert\\
\vert\\
\vdots\\
\vert%
\end{array}
\begin{array}
[c]{c}%
h_{1,l_{1}}\\
h_{2,l_{2}}\\
\vdots\\
h_{n-k,l_{n-k}}%
\end{array}
. \label{eq:general-stab-generators}%
\end{equation} 
Each entry $h_{i,j}$ is an $n$-qubit Pauli operator and $l_{i}$ is the degree
of generator$~h_{i}$ (in general, the degrees $l_{i}$ can be different from
each other). We obtain the other generators of the code by shifting the above
generators to the right by multiples of $n$ qubits. (In the above, note that
the entries $h_{1,l_{1}}$, $h_{2,l_{2}}$, \ldots, $h_{n-k,l_{n-k}}$ are not required
to be in the same column, but we have written it in the above way for convenience.)

We select the first quantum convolutional code from Figure~1 of Ref.~\cite{GR07}\ as our running
example for this paper. This code has
the following two generators:%
\begin{equation}%
\begin{array}
[c]{cc}%
h_{1} & =\\
h_{2} & =
\end{array}%
\begin{array}
[c]{c}%
XXXX\\
ZZZZ
\end{array}
\left\vert
\begin{array}
[c]{c}%
XXIX\\
ZZIZ
\end{array}
\right\vert \left.
\begin{array}
[c]{c}%
IXII\\
IZII
\end{array}
\right\vert
\begin{array}
[c]{c}%
IIXX\\
IIZZ
\end{array}
, \label{eq:running-example}%
\end{equation}
with $n=4$ and $n-k=2$, implying that the code encodes $k=2$ information
qubits for every four physical qubits. Observe that the above generators
commute with each other and with the generators resulting from\ all possible
four-qubit shifts of the above generators.

\section{The Proposed Encoding Algorithm}

\begin{figure}
[ptb]
\begin{center}
\includegraphics[
natheight=4.333600in,
natwidth=5.586700in,
height=3.2007in,
width=3.2396in
]%
{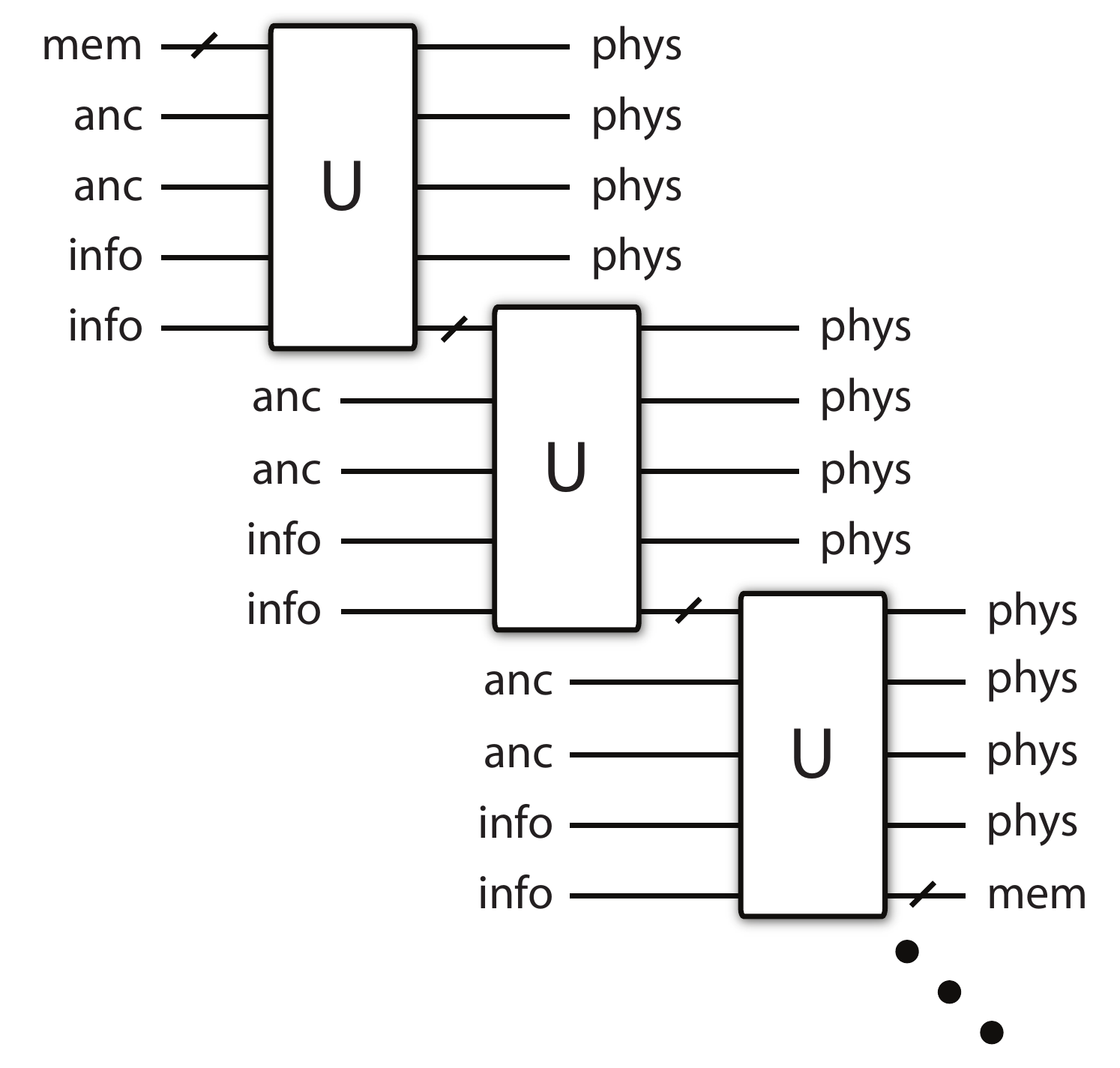}%
\caption{The encoder $U$ for a quantum convolutional code that has four physical
qubits for every two information qubits. The encoder $U$ acts on $m$ memory
qubits, two ancilla qubits, and two information qubits to produce four output
physical qubits to be sent over the channel and $m$ output memory qubits to be
fed into the next round of encoding.}%
\label{fig:conv-encoder}%
\end{center}
\end{figure}

Figure~\ref{fig:conv-encoder}\ depicts an example of an encoder for a quantum
convolutional code. The encoder depicted there can encode our running example
in (\ref{eq:running-example}) that has four physical qubits for every two
information qubits. More generally, a convolutional encoder acts on some
number$~m$ of memory qubits, $n-k$ ancilla qubits, and $k$~information qubits,
and it produces $n$~output physical qubits and $m$~output memory qubits to be
fed into the next round of encoding.

For our example in (\ref{eq:running-example}), the unencoded qubit stream
might have the following form:%
\begin{equation}
\left\vert 0\right\rangle \left\vert 0\right\rangle \left\vert \phi
_{1}\right\rangle \left\vert \phi_{2}\right\rangle \left\vert 0\right\rangle
\left\vert 0\right\rangle \left\vert \phi_{3}\right\rangle \left\vert \phi
_{4}\right\rangle \cdots, \label{eq:unencoded-stream}%
\end{equation}
so that an ancilla qubit appears as every first and second qubit and an
information qubit appears as every third and fourth qubit (generally, these
information qubits can be entangled with each other and even with an
inaccessible reference system, but we write them as product states for
simplicity). A particular set of stabilizer generators for the unencoded qubit stream in
(\ref{eq:unencoded-stream}) is as follows (along with all of their four-qubit
shifts):%
\begin{equation}%
\begin{array}
[c]{c}%
ZIII\\
IZII
\end{array}
\left\vert
\begin{array}
[c]{c}%
IIII\\
IIII
\end{array}
\right\vert \left.
\begin{array}
[c]{c}%
IIII\\
IIII
\end{array}
\right\vert
\begin{array}
[c]{c}%
IIII\\
IIII
\end{array}
, \label{eq:running-example-unencoded}%
\end{equation}
so that the states in (\ref{eq:unencoded-stream}) are in the simultaneous
$+1$-eigenspace of the above operators and all of their four-qubit shifts.

The objective of the convolutional encoder is to transform these
\textquotedblleft unencoded\textquotedblright\ Pauli $Z$ operators to the
encoded stabilizer generators in (\ref{eq:running-example}). That is, it
should be some Clifford transformation\footnote{A Clifford transformation is a
unitary operator that preserves the Pauli group under unitary conjugation.} of
the following form:%
\begin{equation}%
\begin{tabular}
[c]{c|cc|cc}%
Mem. & \multicolumn{2}{|c|}{Anc.} & \multicolumn{2}{|c}{Info.}\\\hline\hline
$I^{\otimes m}$ & $Z$ & $I$ & $I$ & $I$\\
$g_{1,1}$ & $I$ & $I$ & $I$ & $I$\\
$g_{1,2}$ & $I$ & $I$ & $I$ & $I$\\
$g_{1,3}$ & $I$ & $I$ & $I$ & $I$\\\hline
$I^{\otimes m}$ & $I$ & $Z$ & $I$ & $I$\\
$g_{2,1}$ & $I$ & $I$ & $I$ & $I$\\
$g_{2,2}$ & $I$ & $I$ & $I$ & $I$\\
$g_{2,3}$ & $I$ & $I$ & $I$ & $I$%
\end{tabular}
\rightarrow%
\begin{tabular}
[c]{cccc|c}%
\multicolumn{4}{c|}{Phys.} & Mem.\\\hline\hline
$X$ & $X$ & $X$ & $X$ & $g_{1,1}$\\
$X$ & $X$ & $I$ & $X$ & $g_{1,2}$\\
$I$ & $X$ & $I$ & $I$ & $g_{1,3}$\\
$I$ & $I$ & $X$ & $X$ & $I^{\otimes m}$\\\hline
$Z$ & $Z$ & $Z$ & $Z$ & $g_{2,1}$\\
$Z$ & $Z$ & $I$ & $Z$ & $g_{2,2}$\\
$I$ & $Z$ & $I$ & $I$ & $g_{2,3}$\\
$I$ & $I$ & $Z$ & $Z$ & $I^{\otimes m}$%
\end{tabular}
\label{eq:example-encoder}%
\end{equation}
where, as a visual aid, we have separated the memory qubits, ancilla qubits,
and information qubits at the input with a vertical bar and we have done the
same for the physical qubits and memory qubits at the output. A horizontal bar
separates the entries of the encoder needed to encode the first generator from
the entries needed to encode the second generator. Each $g_{i,j}$ is a Pauli
operator acting on some number~$m$ of memory qubits---these operators should
be consistent with the input-output commutation relations of the encoder (more
on this later). We stress that the above input-output relations only partially
specify the encoder such that it produces a code with the stabilizer
generators in (\ref{eq:running-example}), and there is still a fair amount of
freedom remaining in the encoding.%

In the general case, a convolutional encoder should transform an unencoded
Pauli $Z$ operator acting on the $i^{\text{th}}$ ancilla qubit to the
$i^{\text{th}}$ stabilizer generator $h_{i}$ in
(\ref{eq:general-stab-generators}). The first application of the encoder $U$
results in an intermediate, unspecified Pauli operator $g_{i,1}$ acting on the
$m$ output memory qubits. The second application of the encoder $U$ results in
an intermediate, unspecified Pauli operator $g_{i,2}$ acting on the $m$ output
memory qubits and so on. The shift invariance of the overall encoding
guarantees that shifts of the unencoded$~Z$ Pauli operators transform to
appropriate shifts of the generators. A convolutional encoder for the code
should perform the following transformation:%
\begin{equation}%
\begin{tabular}
[c]{c|c|c}%
Mem. & \multicolumn{1}{|c|}{Anc.} & Info.\\\hline\hline
$I^{\otimes m}$ & $Z_{1}$ & $I^{\otimes k}$\\
$g_{1,1}$ & \multicolumn{1}{|c|}{$I^{\otimes n-k}$} & $I^{\otimes k}$\\
$\vdots$ & \multicolumn{1}{|c|}{$\vdots$} & $\vdots$\\
$g_{1,l_{1}-2}$ & \multicolumn{1}{|c|}{$I^{\otimes n-k}$} & $I^{\otimes k}$\\
$g_{1,l_{1}-1}$ & \multicolumn{1}{|c|}{$I^{\otimes n-k}$} & $I^{\otimes k}%
$\\\hline
$I^{\otimes m}$ & $Z_{2}$ & $I^{\otimes k}$\\
$g_{2,1}$ & \multicolumn{1}{|c|}{$I^{\otimes n-k}$} & $I^{\otimes k}$\\
$\vdots$ & \multicolumn{1}{|c|}{$\vdots$} & $\vdots$\\
$g_{2,l_{2}-2}$ & \multicolumn{1}{|c|}{$I^{\otimes n-k}$} & $I^{\otimes k}$\\
$g_{2,l_{2}-1}$ & \multicolumn{1}{|c|}{$I^{\otimes n-k}$} & $I^{\otimes k}%
$\\\hline
$\vdots$ & \multicolumn{1}{|c|}{$\vdots$} & $\vdots$\\\hline
$I^{\otimes m}$ & $Z_{s}$ & $I^{\otimes k}$\\
$g_{s,1}$ & \multicolumn{1}{|c|}{$I^{\otimes n-k}$} & $I^{\otimes k}$\\
$\vdots$ & \multicolumn{1}{|c|}{$\vdots$} & $\vdots$\\
$g_{s,l_{s}-2}$ & \multicolumn{1}{|c|}{$I^{\otimes n-k}$} & $I^{\otimes k}$\\
$g_{s,l_{s}-1}$ & \multicolumn{1}{|c|}{$I^{\otimes n-k}$} & $I^{\otimes k}$%
\end{tabular}
\rightarrow%
\begin{tabular}
[c]{c|c}%
Phys. & Mem.\\\hline\hline
$h_{1,1}$ & $g_{1,1}$\\
$h_{1,2}$ & $g_{1,2}$\\
$\vdots$ & $\vdots$\\
$h_{1,l_{1}-1}$ & $g_{1,l_{1}-1}$\\
$h_{1,l_{1}}$ & $I^{\otimes m}$\\\hline
$h_{2,1}$ & $g_{2,1}$\\
$h_{2,2}$ & $g_{2,2}$\\
$\vdots$ & $\vdots$\\
$h_{2,l_{2}-1}$ & $g_{2,l_{2}-1}$\\
$h_{2,l_{2}}$ & $I^{\otimes m}$\\\hline
$\vdots$ & $\vdots$\\\hline
$h_{s,1}$ & $g_{s,1}$\\
$h_{s,2}$ & $g_{s,2}$\\
$\vdots$ & $\vdots$\\
$h_{s,l_{s}-1}$ & $g_{s,l_{s}-1}$\\
$h_{s,l_{s}}$ & $I^{\otimes m}$%
\end{tabular}
\label{eq:general-encoder}%
\end{equation}
where $m$ is some unspecified number of memory qubits, $k$ is the number of
information qubits, $n-k$ is the number of ancilla qubits, and we make the
abbreviation $s\equiv n-k$. Again, the above transformation only partially
specifies the encoding. Also, note that it is not necessary for the $g_{i,j}$ operators
to be independent of one another and we address this point later on.

\subsection{Consistency of Commutation Relations}

A fundamental property of any valid Clifford transformation is that it
preserves commutation relations. That is, the input commutation relations
should be consistent with the output commutation relations. So, for all $1\leq
i\leq n-k$ and $1\leq j\leq l_{i}-1$, the entries $g_{i,j}$ are $m$-qubit
Pauli operators that are unspecified above, but they should be chosen in such
a way that the input-output commutation relations are consistent. That this
consistency is possible follows from the fact that the stabilizer generators
in (\ref{eq:general-stab-generators}) form a valid quantum convolutional code
according to Definition~\ref{def:QCC}, and it is the content of our first theorem.

\begin{theorem}
[Consistency of Commutation Relations]Suppose the stabilizer generators in
(\ref{eq:general-stab-generators}) form a valid quantum convolutional code.
Then there exists a set of Pauli operators $g_{i,j}$ for $1\leq i\leq n-k$ and
$1\leq j\leq l_{i}$ such that the commutation relations on the
LHS\ of~(\ref{eq:general-encoder}) are consistent with those on the
RHS\ of~(\ref{eq:general-encoder}).
\end{theorem}

\begin{proof}
Let $g_{i,j}\odot g_{k,l}$ be a function that equals one if $g_{i,j}$ and
$g_{k,l}$ anticommute and zero if they commute. By inspecting the
transformation in (\ref{eq:general-encoder}), several commutation relations
should be satisfied. First, for all $i,i^{\prime}\in\left\{  1,2,\cdots
,n-k\right\}  $ and for all $j^{\prime}\in\left\{  1,\cdots,l_{i^{\prime}%
}-1\right\}  $:%
\[
g_{i,1}\odot g_{i^{\prime},j^{\prime}}=h_{i,1}\odot h_{i^{\prime},j^{\prime}%
},
\]
because the first row of each block on the LHS\ of (\ref{eq:general-encoder})
commutes with all other rows and for consistency, the RHS\ of the
corresponding rows should commute as well. Next, for all $\,i,i^{\prime}%
\in\{1,2,\cdots,n-k\}$, $j\in\{1,\cdots,l_{i}-2\}$, and $j^{\prime}%
\in\{1,\cdots,l_{i^{\prime}}-2\}$:%
\[
g_{i,j}\odot g_{i^{\prime},j^{\prime}}=\left(  h_{i,j+1}\odot h_{i^{\prime
},j^{\prime}+1}\right)  +\left(  g_{i,j+1}\odot g_{i^{\prime},j^{\prime}%
+1}\right)  ,
\]
because the commutation relations between any of the second to second-to-last
rows in the same or different blocks on the LHS\ of (\ref{eq:general-encoder})
should be consistent with those of the corresponding rows on the RHS. Finally,
for all $i,i^{\prime}\in\{1,2,\cdots,n-k\}$ and $j\in\{1,\cdots,l_{i}-1\}$:%
\[
g_{i,j}\odot g_{i^{\prime},l_{i^{\prime}}-1}=h_{i,j+1}\odot h_{i^{\prime
},l_{i^{\prime}}},
\]
because the commutation relations between the last row of each block and any
other row on the LHS\ of (\ref{eq:general-encoder}) should be consistent with
those of the corresponding rows on the RHS.

If we start from the first row of any block in (\ref{eq:general-encoder}), a
forward commutativity propagation imposes the following equality (WLOG suppose
$j\geq j^{\prime}$):
\begin{equation}
g_{i,j}\odot g_{i^{\prime},j^{\prime}}=\sum_{k=1}^{\min\{(l_{i}%
-j),(l_{i^{\prime}}-j^{\prime})\}}h_{i,j+k}\odot h_{i^{\prime},j^{\prime}+k},
\label{cons1}%
\end{equation}
and if we start from the last row of any block in (\ref{eq:general-encoder}),
a backward commutativity propagation imposes the following equality:%
\begin{equation}
g_{i,j}\odot g_{i^{\prime},j^{\prime}}=\sum_{k=0}^{j^{\prime}-1}h_{i,j-k}\odot
h_{i^{\prime},j^{\prime}-k}. \label{cons2}%
\end{equation}
By adding the RHS of~(\ref{cons1}) and~(\ref{cons2}), we obtain the following
equality:%
\begin{align*}
&  \sum_{k=1}^{\min\{(l_{i}-j),(l_{i^{\prime}}-j^{\prime})\}}h_{i,j+k}\odot
h_{i^{\prime},j^{\prime}+k}+\sum_{k=0}^{j^{\prime}-1}h_{i,j-k}\odot
h_{i^{\prime},j^{\prime}-k}\\
&  =\sum_{k=1}^{\min\{{(l_{i}-j+j^{\prime}),l_{i^{\prime}}}\}}%
h_{i,k+j-j^{\prime}}\odot h_{i^{\prime},k}\\
&  =\left(  D^{j-j^{\prime}}h_{i}\right)  \odot h_{i^{\prime}},
\end{align*}
where we have introduced the delay operator $D$ from
Refs.~\cite{PhysRevLett.91.177902,ieee2007forney}. Finally, due to the
commutativity constraints for the generators of a valid quantum convolutional
code, we obtain the following equality:%
\begin{align*}
\sum_{k=1}^{\min\{{(l_{i}-j+j^{\prime}),l_{i^{\prime}}}\}}h_{i,k+j-j^{\prime}%
}\odot h_{i^{\prime},k}  &  =\left(  D^{j-j^{\prime}}h_{i}\right)  \odot
h_{i^{\prime}}\\
&  =0.
\end{align*}
Therefore, the RHS of equations in~(\ref{cons1}) and~(\ref{cons2}) are the
same, and the different constraints imposed by the encoder on the commutation
relations of $g_{i,j}$ and $g_{i^{\prime},j^{\prime}}$ are consistent.
\end{proof}

The next section shows how to choose the operators $g_{i,j}$\ for the memory
qubits such that they are consistent while also acting on a minimal number of
memory qubits.

\subsection{Memory Commutativity Matrix}

\label{sec:mem-comm}In our running example in (\ref{eq:running-example}) and
(\ref{eq:example-encoder}), we did not specify how to choose the Pauli
operators $g_{i,j}$ acting on the memory qubits. It would be ideal to choose
them so that they are consistent with the input-output commutation relations
of the transformation in (\ref{eq:example-encoder}), and also so that they act
on a minimal number of memory qubits. In this way, we can determine a
minimal-memory encoder for the particular stabilizer generators
in~(\ref{eq:running-example}).

As stated earlier, any valid Clifford transformation
preserves commutation relations. That is, if two input Pauli operators
commute, then the corresponding output Pauli operators should also
commute\ (and similarly, two outputs should anticommute if their corresponding
inputs anticommute). So, consider that the first two input rows in
(\ref{eq:example-encoder}) commute. Then the two output rows should commute as
well, and in order for this to happen, $g_{1,1}$ and $g_{1,2}$ should commute
because $XXXX$ and $XXIX$ commute. For a different case, observe that the
first and fifth input rows commute, and for consistency, the first and fifth
output rows should commute. Thus, $g_{1,1}$ and $g_{2,1}$ should commute
because $XXXX$ and $ZZZZ$ already commute. We can continue in this manner and
enumerate all of the commutation relations for the memory operators $g_{i,j}$
simply by ensuring that the input-output commutation relations in
(\ref{eq:example-encoder}) are consistent:%
\begin{align}
\left[  g_{1,1},g_{1,2}\right]   &  =\left[  g_{1,1},g_{1,3}\right]  =\left[
g_{1,1},g_{2,1}\right] \nonumber\\
&  =\left\{  g_{1,1},g_{2,2}\right\}  =\left\{  g_{1,1},g_{2,3}\right\}  =0,\label{eq:comm-rel-1}\\
\left[  g_{1,2},g_{1,3}\right]   &  =\left\{  g_{1,2},g_{2,1}\right\}
=\left\{  g_{1,2},g_{2,2}\right\} \nonumber\\
&  =\left[  g_{1,2},g_{2,3}\right]  =0,\\
\left\{  g_{1,3},g_{2,1}\right\}   &  =\left[  g_{1,3},g_{2,2}\right]
=\left[  g_{1,3},g_{2,3}\right]  =0,\\
\left[  g_{2,1},g_{2,2}\right]   &  =\left[  g_{2,1},g_{2,3}\right]  =0,\\
\left[  g_{2,2},g_{2,3}\right]   &  =0, \label{eq:comm-rel-5}%
\end{align}
where $\left[  A,B\right]  \equiv AB-BA$ is the commutator and $\left\{
A,B\right\}  \equiv AB+BA$ is the anticommutator. In determining some of the
later commutation relations, we need to rely on earlier found ones.

Our objective now is to determine the minimal number of memory qubits on which
the operators $g_{i,j}$ should act in order for the transformation in
(\ref{eq:example-encoder}) to be consistent with the commutation relations in
(\ref{eq:comm-rel-1}-\ref{eq:comm-rel-5}). To this end, it is helpful to write
the above commutation relations as entries in a square binary-valued
matrix~$\Omega$, that we refer to as the \textquotedblleft memory
commutativity matrix.\textquotedblright

\begin{definition}
[Memory Commutativity Matrix]The memory commutativity matrix~$\Omega$
corresponding to an encoder of the form in (\ref{eq:general-encoder}) for a
set of stabilizer generators has its entries equal to%
\begin{equation}
\left[  \Omega\right]  _{\left(  i,j\right)  ,\left(  k,l\right)  }\equiv
g_{i,j}\odot g_{k,l}, \label{eq:mem-comm-matrix}%
\end{equation}
where we think of the double indices $\left(  i,j\right)  $ and $\left(
k,l\right)  $ as single indices for the matrix elements of $\Omega$, $g_{i,j}$
and $g_{k,l}$ are all of the Pauli operators in (\ref{eq:general-encoder})
acting on the memory qubits, and $g_{i,j}\odot g_{k,l}$ is a function that
equals one if $g_{i,j}$ and $g_{k,l}$ anticommute and zero if they commute
(implying that $\Omega$ is a symmetric matrix).
\end{definition}

For our running example in (\ref{eq:running-example}),
(\ref{eq:example-encoder}), and (\ref{eq:comm-rel-1}-\ref{eq:comm-rel-5}), the
memory commutativity matrix~$\Omega$ is equal to%
\begin{equation}%
\begin{bmatrix}
0 & 0 & 0 & 0 & 1 & 1\\
0 & 0 & 0 & 1 & 1 & 0\\
0 & 0 & 0 & 1 & 0 & 0\\
0 & 1 & 1 & 0 & 0 & 0\\
1 & 1 & 0 & 0 & 0 & 0\\
1 & 0 & 0 & 0 & 0 & 0
\end{bmatrix}
, \label{eq:mem-comm-mat-example}%
\end{equation}
if we take the ordering $g_{1,1}$, $g_{1,2}$, $g_{1,3}$, $g_{2,1}$, $g_{2,2}$,
$g_{2,3}$ and consider the commutation relations found in (\ref{eq:comm-rel-1}%
-\ref{eq:comm-rel-5}).

The memory commutativity matrix captures commutation relations between Pauli
matrices, and our objective is to determine the minimal number of memory
qubits on which the memory operators should act in order to be consistent with
the above commutation relations. This leads us to our next theorem:

\begin{theorem}
[Minimal-Memory Encoder]\label{thm:min-mem}For a given memory commutativity
matrix~$\Omega$, the minimal number~$m$ of memory qubits needed for an encoder
is equal to%
\[
m=\dim\left(  \Omega\right)  -\frac{1}{2}\text{\emph{rank}}\left(
\Omega\right)  .
\]

\end{theorem}

\begin{proof}
To prove this theorem, we can exploit ideas from the theory of
entanglement-assisted quantum error correction~\cite{BDH06}, after realizing
that finding the minimal number of memory qubits on which the memory operators
should act is related to finding the minimal number of ebits required in an
entanglement-assisted quantum code. In particular, by the symplectic
Gram-Schmidt procedure outlined in
Refs.~\cite{BDH06,arx2008wildeOEA,PhysRevA.79.062322}, there exists a sequence
of full-rank matrices acting by conjugation on the memory commutativity
matrix~$\Omega$ that reduces it to the following standard form:%
\begin{equation}
\Omega_{0}\equiv\bigoplus\limits_{k=1}^{c}%
\begin{bmatrix}
0 & 1\\
1 & 0
\end{bmatrix}
\oplus\bigoplus\limits_{l=1}^{d}\left[  0\right]  ,
\label{eq:standard-form-mem-comm}%
\end{equation}
such that $2c+d=\dim\left(  \Omega\right)  $ for some integers $c,d\geq0$. Let
$G$ denote this sequence of operations. Observe that $\dim\left(
\Omega\right)  =\dim\left(  \Omega_{0}\right)  $ and rank$\left(
\Omega\right)  =\ $rank$\left(  \Omega_{0}\right)  $ because this sequence~$G$
of operations is full rank. Furthermore, it holds that rank$\left(  \Omega
_{0}\right)  =2c$ because the rank of a direct sum is the sum of the
individual matrix ranks. Observe that the Pauli operators $X_{1}$, $Z_{1}$,
\ldots, $X_{c}$, $Z_{c}$ and $Z_{c+1}$, \ldots, $Z_{c+d}$ acting on $c+d$
qubits have the same commutativity matrix as the standard form given in
(\ref{eq:standard-form-mem-comm}), and furthermore, these operators are
minimal, in the sense that there is no set of operators acting on fewer than
$c+d$ qubits that could satisfy the commutation relations in
(\ref{eq:standard-form-mem-comm}). We then perform the inverse $G^{-1}$ on the
operators $X_{1}$, $Z_{1}$, \ldots, $X_{c}$, $Z_{c}$ and $Z_{c+1}$, \ldots,
$Z_{c+d}$, producing a set of memory operators $g_{i,j}$ that are consistent
with the commutation relations in (\ref{eq:mem-comm-matrix}), ensuring that
the encoder is valid, while acting on the minimal number of memory qubits
possible. The resulting number $m$ of memory qubits is then $m=c+d$, or
equivalently,%
\[
m=\dim\left(  \Omega\right)  -\frac{1}{2}\text{rank}\left(  \Omega\right)  ,
\]
because $\dim\left(  \Omega\right)  =2c+d$ and rank$\left(  \Omega\right)
=2c$~\cite{arx2008wildeOEA,PhysRevA.79.062322}.
\end{proof}

We can apply the above theorem to our running example in
(\ref{eq:running-example}) and (\ref{eq:example-encoder}). The rank of the
matrix in (\ref{eq:mem-comm-mat-example}) is full (equal to six), implying
that $c=3$ and the minimal number of memory qubits to encode the generators in
(\ref{eq:running-example}) is three qubits. Indeed, the standard form of the
memory commutativity matrix is%
\[%
\begin{bmatrix}
0 & 1\\
1 & 0
\end{bmatrix}
\oplus%
\begin{bmatrix}
0 & 1\\
1 & 0
\end{bmatrix}
\oplus%
\begin{bmatrix}
0 & 1\\
1 & 0
\end{bmatrix}
.
\]
A set of Pauli operators with commutation relations corresponding to this
standard form is $X_{1}$, $Z_{1}$, $X_{2}$, $Z_{2}$, $X_{3}$, and $Z_{3}$. We
can multiply these Pauli operators together to produce the generators
$g_{1,1}=XIX$, $g_{1,2}=IIX$, $g_{1,3}=IZI$, $g_{2,1}=ZXZ$, $g_{2,2}=IIZ$, and
$g_{2,3}=ZII$ with a commutativity matrix equivalent to that in
(\ref{eq:mem-comm-mat-example}). We can then use these generators as memory
operators for the encoder in (\ref{eq:example-encoder}), producing the
following valid minimal-memory convolutional encoder for the stabilizer
generators in (\ref{eq:running-example}):%
\begin{equation}%
\begin{tabular}
[c]{c|c|c}%
Mem. & \multicolumn{1}{|c|}{Anc.} & Info.\\\hline\hline
$III$ & $ZI$ & $II$\\
$XIX$ & $II$ & $II$\\
$IIX$ & $II$ & $II$\\
$IZI$ & $II$ & $II$\\\hline
$III$ & $IZ$ & $II$\\
$ZXZ$ & $II$ & $II$\\
$IIZ$ & $II$ & $II$\\
$ZII$ & $II$ & $II$%
\end{tabular}
\rightarrow%
\begin{tabular}
[c]{c|c}%
\multicolumn{1}{c|}{Phys.} & Mem.\\\hline\hline
$XXXX$ & $XIX$\\
$XXIX$ & $IIX$\\
$IXII$ & $IZI$\\
$IIXX$ & $III$\\\hline
$ZZZZ$ & $ZXZ$\\
$ZZIZ$ & $IIZ$\\
$IZII$ & $ZII$\\
$IIZZ$ & $III$%
\end{tabular}
. \label{eq:run-example-encoder-partial}%
\end{equation}

Once we have determined the transformation that the encoder should perform,
there is an algorithm for determining an encoder with polynomial
depth~\cite{BFG06}. There are many encoders which implement the transformation
in~(\ref{eq:run-example-encoder-partial}). In order to specify a particular
encoder in full, one would need to \textquotedblleft
complete\textquotedblright\ the above transformation by determining six
additional input-output relations that are independent of the other
input-output relations, so that the resulting 14 input-output relations form a
basis for the Pauli group acting on seven qubits.

\section{Other representations of a code}

We can find other representations
of a quantum convolutional code by multiplying stabilizer generators together
or by delaying some of them. In this section, we analyze the impact of
these operations on the minimal memory requirements for encoders, and we
propose an algorithm to
find a minimal-memory encoder among all the representations of a given code.

\subsection{Multiplication of stabilizers}

Suppose we obtain another set of stabilizer generators (say, $S^{\prime}$) for the same code
specified in~(\ref{eq:general-stab-generators}), by multiplying one
stabilizer by another. WLOG, suppose that the first stabilizer generator $h_{1}$ is
multiplied by second stabilizer generator $h_{2}$, and suppose that $l_{1}>l_{2}$.
As a result,
only the rows of transformation corresponding to the second stabilizer generator for
$S^{\prime}$ (the rows in the second block of the transformation) are different from the
rows of the transformation corresponding to the original set of stabilizers
in~(\ref{eq:general-stab-generators}). In the following, we write the rows of the
transformation corresponding to the first and second stabilizer for
$S^{\prime}$:%

\begin{equation}%
\begin{tabular}
[c]{c|c|c}%
Mem. & \multicolumn{1}{|c|}{Anc.} & Info.\\\hline\hline
$I^{\otimes m}$ & $Z_{1}$ & $I^{\otimes k}$\\
$g_{1,1}$ & \multicolumn{1}{|c|}{$I^{\otimes n-k}$} & $I^{\otimes k}$\\
$\vdots$ & \multicolumn{1}{|c|}{$\vdots$} & $\vdots$\\
$g_{1,l_{1}-2}$ & \multicolumn{1}{|c|}{$I^{\otimes n-k}$} & $I^{\otimes k}$\\
$g_{1,l_{1}-1}$ & \multicolumn{1}{|c|}{$I^{\otimes n-k}$} & $I^{\otimes k}%
$\\\hline
$I^{\otimes m}$ & $Z_{2}$ & $I^{\otimes k}$\\
$g_{2,1}$ & \multicolumn{1}{|c|}{$I^{\otimes n-k}$} & $I^{\otimes k}$\\
$\vdots$ & \multicolumn{1}{|c|}{$\vdots$} & $\vdots$\\
$g_{2,l_{2}-1}$ & \multicolumn{1}{|c|}{$I^{\otimes n-k}$} & $I^{\otimes k}$\\
$g_{2,l_{2}}$ & \multicolumn{1}{|c|}{$I^{\otimes n-k}$} & $I^{\otimes k}$\\
$\vdots$ & \multicolumn{1}{|c|}{$\vdots$} & $\vdots$\\
$g_{2,l_{1}-2}$ & \multicolumn{1}{|c|}{$I^{\otimes n-k}$} & $I^{\otimes k}$\\
$g_{2,l_{1}-1}$ & \multicolumn{1}{|c|}{$I^{\otimes n-k}$} & $I^{\otimes k}%
$\\\hline
\end{tabular}
\rightarrow%
\begin{tabular}
[c]{c|c}%
Phys. & Mem.\\\hline\hline
$h_{1,1}$ & $g_{1,1}$\\
$h_{1,2}$ & $g_{1,2}$\\
$\vdots$ & $\vdots$\\
$h_{1,l_{1}-1}$ & $g_{1,l_{1}-1}$\\
$h_{1,l_{1}}$ & $I^{\otimes m}$\\\hline
$h_{2,1}\times h_{1,1}$ & $g_{2,1}$\\
$h_{2,2}\times h_{1,2}$ & $g_{2,2}$\\
$\vdots$ & $\vdots$\\
$h_{2,l_{2}}\times h_{1,l_{2}}$ & $g_{2,l_{2}}$\\
$h_{1,l_{2}+1}$ & $g_{2,l_{2}+1}$\\
$\vdots$ & $\vdots$\\
$h_{1,l_{1}-1}$ & $g_{1,l_{1}-1}$\\
$h_{1,l_{1}}$ & $I^{\otimes m}$\\\hline
\end{tabular}
\label{eq:other-representation-mul-1}%
\end{equation}

The RHS of the last row of the first and second block
in~(\ref{eq:other-representation-mul-1}) are the same. So we deduce that the
memory states of $g_{1,l_{1}-1}$ and $g_{2,l_{1}-1}$ are the same as well. Thus
we can omit the last row of the second block and exchange $g_{2,l_{1}-1}$ by
$g_{1,l_{1}-1}$ in the transformation. By proceeding in the transformation and
omitting repetitive rows, it will turn into the following transformation:
\begin{equation}%
\begin{tabular}
[c]{c|c|c}%
Mem. & \multicolumn{1}{|c|}{Anc.} & Info.\\\hline\hline
$I^{\otimes m}$ & $Z_{1}$ & $I^{\otimes k}$\\
$g_{1,1}$ & \multicolumn{1}{|c|}{$I^{\otimes n-k}$} & $I^{\otimes k}$\\
$\vdots$ & \multicolumn{1}{|c|}{$\vdots$} & $\vdots$\\
$g_{1,l_{2}-1}$ & \multicolumn{1}{|c|}{$I^{\otimes n-k}$} & $I^{\otimes k}$\\
$g_{1,l_{2}}$ & \multicolumn{1}{|c|}{$I^{\otimes n-k}$} & $I^{\otimes k}$\\
$\vdots$ & \multicolumn{1}{|c|}{$\vdots$} & $\vdots$\\
$g_{1,l_{1}-2}$ & \multicolumn{1}{|c|}{$I^{\otimes n-k}$} & $I^{\otimes k}$\\
$g_{1,l_{1}-1}$ & \multicolumn{1}{|c|}{$I^{\otimes n-k}$} & $I^{\otimes k}%
$\\\hline
$I^{\otimes m}$ & $Z_{2}$ & $I^{\otimes k}$\\
$g_{2,1}$ & \multicolumn{1}{|c|}{$I^{\otimes n-k}$} & $I^{\otimes k}$\\
$\vdots$ & \multicolumn{1}{|c|}{$\vdots$} & $\vdots$\\
$g_{2,l_{2}-1}$ & \multicolumn{1}{|c|}{$I^{\otimes n-k}$} & $I^{\otimes k}$%
\end{tabular}
\rightarrow%
\begin{tabular}
[c]{c|c}%
Phys. & Mem.\\\hline\hline
$h_{1,1}$ & $g_{1,1}$\\
$h_{1,2}$ & $g_{1,2}$\\
$\vdots$ & $\vdots$\\
$h_{1,l_{2}}$ & $g_{1,l_{2}}$\\
$h_{1,{l_{2}+1}}$ & $g_{1,{l_{2}+1}}$\\
$\vdots$ & $\vdots$\\
$h_{1,l_{1}-1}$ & $g_{1,l_{1}-1}$\\
$h_{1,l_{1}}$ & $I^{\otimes m}$\\\hline
$h_{2,1}\times h_{1,1}$ & $g_{2,1}$\\
$h_{2,2}\times h_{1,2}$ & $g_{2,2}$\\
$\vdots$ & $\vdots$\\
$h_{2,l_{2}}\times h_{1,l_{2}}$ & $g_{1,l_{2}}$%
\end{tabular}
\label{eq:other-representation-mul-2}%
\end{equation}
By multiplying the first row of the first block by the first row of the second block,
the second row of the first block by the second row of the second block, \ldots, and
the $l_{2}^{\text{th}}$
row of the first block by the $l_{2}^{\text{th}}$ row of the second block
in~(\ref{eq:other-representation-mul-2}), we obtain the following transformation:
\begin{equation}%
\begin{tabular}
[c]{c|cccc|c}%
Mem. & \multicolumn{4}{|c|}{Anc.} & Info.\\\hline\hline
$I^{\otimes m}$ & \multicolumn{4}{|c}{$Z_{1}$} & $I^{\otimes k}$\\
$g_{1,1}$ & \multicolumn{4}{|c|}{$I^{\otimes n-k}$} & $I^{\otimes k}$\\
$\vdots$ & \multicolumn{4}{|c|}{$\vdots$} & $\vdots$\\
$g_{1,l_{2}-1}$ & \multicolumn{4}{|c|}{$I^{\otimes n-k}$} & $I^{\otimes k}$\\
$g_{1,l_{2}}$ & \multicolumn{4}{|c|}{$I^{\otimes n-k}$} & $I^{\otimes k}$\\
$\vdots$ & \multicolumn{4}{|c|}{$\vdots$} & $\vdots$\\
$g_{1,l_{1}-2}$ & \multicolumn{4}{|c|}{$I^{\otimes n-k}$} & $I^{\otimes k}$\\
$g_{1,l_{1}-1}$ & \multicolumn{4}{|c|}{$I^{\otimes n-k}$} & $I^{\otimes k}%
$\\\hline
$I^{\otimes m}$ & \multicolumn{4}{|c|}{$Z_{2}$} & $I^{\otimes k}$\\
$g_{1,1} \times g_{2,1}$ & \multicolumn{4}{|c|}{$I^{\otimes n-k}$} & $I^{\otimes k}%
$\\
$\vdots$ & \multicolumn{4}{|c|}{$\vdots$} & $\vdots$\\
$g_{1,l_{2}-1}\times g_{2,l_{2}-1}$ & \multicolumn{4}{|c|}{$I^{\otimes n-k}$} &
$I^{\otimes k}$%
\end{tabular}
\rightarrow%
\begin{tabular}
[c]{c|c}%
Phys. & Mem.\\\hline\hline
$h_{1,1}$ & $g_{1,1}$\\
$h_{1,2}$ & $g_{1,2}$\\
$\vdots$ & $\vdots$\\
$h_{1,l_{2}}$ & $g_{1,l_{2}}$\\
$h_{1,{l_{2}+1}}$ & $g_{1,{l_{2}+1}}$\\
$\vdots$ & $\vdots$\\
$h_{1,l_{1}-1}$ & $g_{1,l_{1}-1}$\\
$h_{1,l_{1}}$ & $I^{\otimes m}$\\\hline
$h_{2,1}$ & $g_{1,1}\times g_{2,1}$\\
$h_{2,2}$ & $g_{1,2}\times g_{2,2}$\\
$\vdots$ & $\vdots$\\
$h_{2,l_{2}}$ & $I^{\otimes m}$%
\end{tabular}
\label{eq:other-representation-mul-3}%
\end{equation}
By comparing the second block of the above transformation with the second block
of~(\ref{eq:general-encoder}), it is clear that if we write the memory
commutativity matrix for \{$\{g_{1,i}\}_{\{i=1,...,l_{i}-1\}}$, $\{g_{1,i}%
.g_{2,i}\}_{\{i=1,...,l_{i}-1\}}$, \ldots, $\{g_{n-k,i}\}_{\{i=1,...,l_{i}-1\}}%
$\} for the new set of stabilizers, the commutativity matrix will be the same
as the one for the original set of stabilizers, and so the minimal amount of memory
will not change. With a similar approach, we can show that in the case that
$l_{1}\leq l_{2}$, the memory commutativity matrix will not change as well.

\subsection{Delay of stabilizers}

Now suppose we obtain a different representation for the code by delaying one of
the stabilizer generators. Suppose WLOG that the first stabilizer is delayed by $j$
frames. Therefore, the encoder should transform the \textquotedblleft
unencoded\textquotedblright\ Pauli $Z$ operator acting on the first ancilla
qubit to the operator $D^{j}(h_{1})$ (as Figure~\ref{fig:delay} shows). Let $g_{1,l_{1}%
},g_{1,l_{1}+1},...,g_{1,l_{1}+j-1}$ denote the first $j$ memory operators in
the Figure~\ref{fig:delay}. Let $\Omega^{\prime}$ denote the memory commutativity
matrix for the new stabilizer set. Hence the encoder should perform the
transformation in (\ref{eq:other-representation-del-1}). (The first block
in~(\ref{eq:other-representation-del-1}) differs from the first block
in~(\ref{eq:general-encoder}) and the others are the same).

Based on the transformation in~(\ref{eq:other-representation-del-1}), we see
that all memory states in $\{g_{1,l_{1}+s},s\in
\{0,1,2,...,j-1\}\}$ commute with all other memory stabilizers. Based on
this fact, we see that for the other memory states, the commutativity
relations in~(\ref{cons1}) still hold:
\begin{equation}
g_{i,j}\odot g_{i^{\prime},j^{\prime}}=\sum_{k=1}^{\text{min}\{(l_{i}-j),
(l_{i^{\prime}}-j^{\prime})\}} h_{i,j+k}\odot h_{i^{\prime},j^{\prime}+k}%
\end{equation}

Therefore $j$ rows and $j$ columns corresponding to $\{g_{1,l_{1}}%
,\ldots,g_{1,l_{1}+j-1}\}$ in the commutativity matrix ($\Omega^{\prime}$) are all
zero, and the other rows and columns are the same as the corresponding rows
and columns in $\Omega$ (the commutativity matrix for the original set of
generators). Hence, the rank of $\Omega^{\prime}$ is the same as the rank of
$\Omega$, but its dimension is equal to $j+\dim(\Omega)$. Therefore,
it requires $j$ more memory qubits.

In summary, by multiplying stabilizer generators by each other, the amount of memory does not
change, but by delaying one of them by $j$ frames, the required memory increases by $j$.

\begin{equation}%
\begin{tabular}
[c]{c|c|c}%
Mem. & Anc. & Info.\\\hline\hline
$I^{\otimes m}$ & $Z_{1}$ & $I^{\otimes k}$\\
$g_{1,l_{1}}$ & $I^{\otimes n-k}$ & $I^{\otimes k}$\\
$\vdots$ & $\vdots$ & $\vdots$\\
$g_{1,l_{1}+j-2}$ & $I^{\otimes n-k}$ & $I^{\otimes k}$\\
$g_{1,l_{1}+j-1}$ & $I^{\otimes n-k}$ & $I^{\otimes k}$\\
$g_{1,1}$ & $I^{\otimes n-k}$ & $I^{\otimes k}$\\
$\vdots$ & $\vdots$ & $\vdots$\\
$g_{1,l_{1}-2}$ & $I^{\otimes n-k}$ & $I^{\otimes k}$\\
$g_{1,l_{1}-1}$ & $I^{\otimes n-k}$ & $I^{\otimes k}$\\\hline
$I^{\otimes m}$ & $Z_{2}$ & $I^{\otimes k}$\\
$g_{2,1}$ & $I^{\otimes n-k}$ & $I^{\otimes k}$\\
$\vdots$ & $\vdots$ & $\vdots$\\
$g_{2,l_{2}-2}$ & $I^{\otimes n-k}$ & $I^{\otimes k}$\\
$g_{2,l_{2}-1}$ & $I^{\otimes n-k}$ & $I^{\otimes k}$\\\hline
$\vdots$ & $\vdots$ & $\vdots$%
\end{tabular}
\rightarrow%
\begin{tabular}
[c]{c|c}%
Phys. & Mem.\\\hline\hline
$I^{\otimes n}$ & $g_{1,l_{1}}$\\
$I^{\otimes n}$ & $g_{1,l_{1}+1}$\\
$\vdots$ & $\vdots$\\
$I^{\otimes n}$ & $g_{1,l_{1}+j-1}$\\
$h_{1,1}$ & $g_{1,1}$\\
$h_{1,2}$ & $g_{1,2}$\\
$\vdots$ & $\vdots$\\
$h_{1,l_{1}-1}$ & $g_{1,l_{1}-1}$\\
$h_{1,l_{1}}$ & $I^{\otimes m}$\\\hline
$h_{2,1}$ & $g_{2,1}$\\
$h_{2,2}$ & $g_{2,2}$\\
$\vdots$ & $\vdots$\\
$h_{2,l_{2}-1}$ & $g_{2,l_{2}-1}$\\
$h_{2,l_{2}}$ & $I^{\otimes m}$\\\hline
$\vdots$ & $\vdots$%
\end{tabular}
, \label{eq:other-representation-del-1}%
\end{equation}

\begin{figure}[ptb]
\begin{center}
\includegraphics[
width=3.4in
]{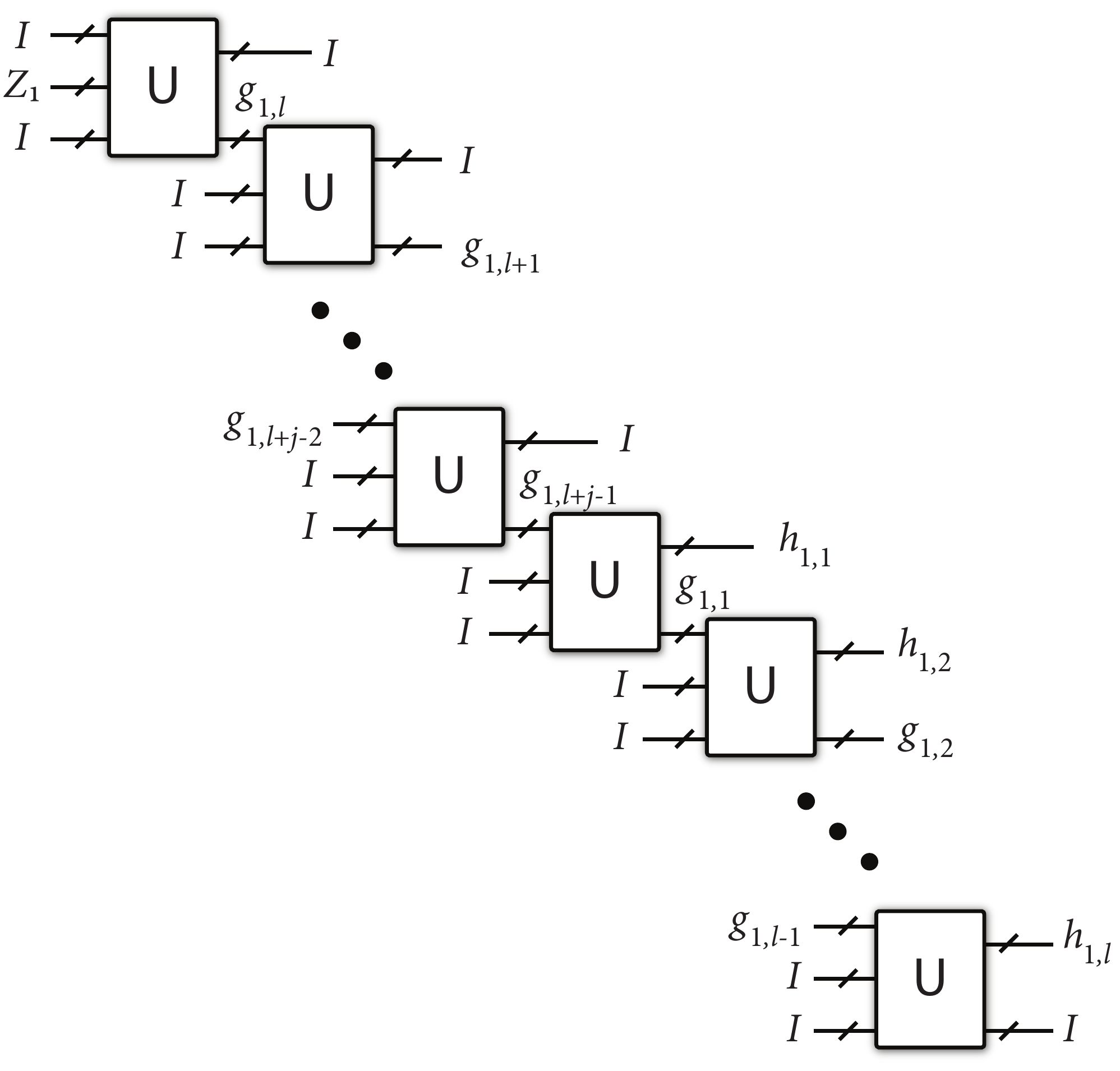}
\end{center}
\caption{If the first stabilizer $h_{1}$ is delayed by $j$ frames, the encoder
should transform the \textquotedblleft unencoded\textquotedblright\ Pauli $Z$
operator acting on the first ancilla qubit to the operator $D^{j}(h_{1})$, as the above
figure shows.}
\label{fig:delay}
\end{figure}

\subsection{Shortening Algorithm}

We can take advantage of the above observations to construct an algorithm
that reduces the minimal memory requirements for a given quantum convolutional code.
First suppose that the first block of each stabilizer generator which acts on
the first $n$ qubits, (i.e., $h_{i,1}, i\in\{1,2,\cdots,n-k\}$) are all
independent of each other. If we find another representation for the code by
multiplying the stabilizers, as we proved, the amount of minimal memory
will be the same, and if we find another representation for the code by
delaying some of the stabilizers, the minimal required memory will be more.
Now suppose that there is a dependence among the generators $h_{i,1}$. Suppose WLOG that
$h_{i,1}$ is equal to $h_{j,1}$, so that by multiplying $h_{i,1}$ with $h_{j,1}$ the
first block of $j^{th}$ stabilizer $h_{j,1}$ becomes equal to the identity.
Therefore, by shifting it one frame to the left the amount of minimal memory
requirement will be decreased. Therefore, for finding the minimal-memory
encoder among all representations of the same code, we should find a
representation in which all the first blocks of the stabilizer generators are
independent of each other. Also we should remove any dependence among the last
blocks of the stabilizer generators as well in order to be confident that the memory
states in~(\ref{eq:general-encoder}) are independent of each other and our
formula is valid. In the next section, we propose an algorithm, that we call
the ``shortening algorithm,'' to be confident that there is no dependence among the first
blocks of the stabilizer generators and also among the last blocks of the generators.
For a given set of stabilizers, first we
should apply the shortening algorithm and then write the transformation
in~(\ref{eq:general-encoder}) for the output stabilizer generators of the
algorithm to find the minimal-memory encoder among all representations of a code.

Algorithm 1 is the algorithm for shortening the stabilizers to be confident
that we are finding the minimal memory requirements among all stabilizer
representations of a given code. There is no dependence among the first blocks and
also last blocks of output stabilizers of the algorithm. The function Subset(S)
returns all the subsets of $S$ except for the empty subset. The complexity of the algorithm is
exponential in $n-k.$ \begin{algorithm}
\caption {Algorithm for shortening generators}
\begin{algorithmic}
\STATE $l_{i,\text{min}}\leftarrow$ minimum degree of $h_i$
\FOR{$i := 1$ to $n-k$}
\STATE $h_i \leftarrow h_i \times D^{-l_{i,min}}$
\ENDFOR
\STATE $\text{DepFound} \leftarrow 1$
\WHILE{$\text{DepFound} = 1$}
\STATE $\text{DepFound}\leftarrow 0$
\FOR{$i := 1$ to $n-k$}
\STATE $m\leftarrow 1$
\FOR{$j := 1$ to $n-k$}
\IF{$i\neq j$ {\bf AND} $l_j \leq l_i$}
\STATE $S_m \leftarrow h_{j,1}$
\STATE $m++$
\ENDIF
\ENDFOR
\STATE $\tilde{S}\leftarrow $ Subset($S$)
\FOR{$y=1$,  $y<2^{m}-1$}
\IF {the product of members of $\tilde{S}_{y}$ is equal to $h_{i,1}$}
\FOR{all $h_{g,1}$ in $\tilde{S}_{y}$}
\STATE $h_{i}\leftarrow h_{i} \times h_{g}$
\ENDFOR
\STATE $h_{i}\leftarrow D^{-1}h_{i}$
\STATE $l_i --$
\STATE $\text{DepFound}\leftarrow 1$
\ENDIF
\ENDFOR
\STATE CLEAR (S);  CLEAR ($\tilde{S}$)
\ENDFOR
\ENDWHILE
\STATE $\text{DepFound} \leftarrow 1$
\WHILE{$\text{DepFound} = 1$}
\STATE $\text{DepFound}\leftarrow 0$
\FOR{$i := 1$ to $n-k$}
\STATE $m\leftarrow 1$
\FOR{$j := 1$ to $n-k$}
\IF{$i\neq j$ {\bf AND } $l_j \leq l_i$}
\STATE $S_m \leftarrow h_{j,l_j}$
\STATE $m++$
\ENDIF
\ENDFOR
\STATE $\tilde{S}\leftarrow $ Subset($S$)
\FOR{$y=1$, $y<2^{m}-1$}
\IF {the product of members of $\tilde{S}_{y}$ is equal to $h_{i,l_i}$}
\FOR{all $h_{g,l_g}$ in $\tilde{S}_{y}$}
\STATE $h_{i}=h_{i} \times D^{l_i-l_g}h_{g}$
\ENDFOR
\STATE $l_i--$
\STATE $\text{DepFound}\leftarrow 1$
\ENDIF
\ENDFOR
\STATE CLEAR (S); CLEAR ($\tilde{S}$)
\ENDFOR
\ENDWHILE
\end{algorithmic}
\end{algorithm}

\section{Catastrophicity}

\label{sec:state-diagram}Although the convolutional encoder in
(\ref{eq:run-example-encoder-partial}) has a minimal number of memory qubits,
it may not necessarily be non-catastrophic (though, we show that it actually
is non-catastrophic in Section~\ref{sec:main-theorem1}). We should ensure that
the encoder is non-catastrophic if the receiver decodes the encoded qubits
with the inverse of the encoder and then exploits the decoding algorithm in
Ref.~\cite{PTO09}\ to correct for errors introduced by a noisy channel. As a
prerequisite for non-catastrophicity, we need to review the notion of a state
diagram for a quantum convolutional encoder.

The state diagram for a quantum convolutional encoder is the most important
tool for analyzing properties such as its distance spectrum and for
determining whether it is catastrophic~\cite{PTO09}. It is similar to the
state diagram for a classical encoder~\cite{V71,book1999conv,McE02}, with an
important exception for the quantum case that incorporates the fact that the
logical operators of a quantum code are unique up to multiplication by the
stabilizer generators. The state diagram allows us to analyze the flow of the
logical operators through the convolutional encoder.

\begin{definition}
[State Diagram]The state diagram for a quantum convolutional encoder is a
directed multigraph with $4^{m}$ vertices that we can think of as
\textquotedblleft memory states,\textquotedblright\ where $m$ is the number of
memory qubits in the encoder. Each memory state corresponds to an $m$-qubit
Pauli operator~$M$ that acts on the memory qubits. We connect two vertices $M$
and $M^{\prime}$ with a directed edge from $M$ to $M^{\prime}$ and label this
edge as $\left(  L,P\right)  $ if the encoder takes the $m$-qubit Pauli
operator~$M$, an $\left(  n-k\right)  $-qubit Pauli operator $S^{z}\in\left\{
I,Z\right\}  ^{n-k}$ acting on the $n-k$ ancilla qubits, and a $k$-qubit Pauli
operator $L$ acting on the information qubits, to an $n$-qubit Pauli operator
$P$ acting on the $n$ physical qubits and an $m$-qubit Pauli
operator~$M^{\prime}$ acting on the $m$ memory qubits:%
\[%
\begin{tabular}
[c]{c|c|c}%
\emph{Mem.} & \emph{Anc.} & \emph{Info.}\\\hline\hline
$M$ & $S^{z}$ & $L$%
\end{tabular}
\ \ \ \ \ \ \underrightarrow{\ \text{encoder\ }}\ \ \ \ \
\begin{tabular}
[c]{c|c}%
\emph{Phys.} & \emph{Mem.}\\\hline\hline
$P$ & $M^{\prime}$%
\end{tabular}
\ .
\]
The labels $L$ and $P$ are the respective logical and physical labels of the edge.
\end{definition}

Observe that the state diagram has $4^{m}$ vertices and $2^{2m+n+k}$ edges
(there are $4^{m}$ memory states, $4^{k}$ logical transitions for $L$, and
$2^{n-k}$ ancilla operators). This is the main reason that it is important to
reduce the size of the encoder's memory---it is related to the complexity of
the decoding algorithm.

We do not explicitly depict the state diagram for our running example because
it would require $4^{3}=64$~vertices and $2^{2\left(  3\right)  +4+2}%
=4096$~edges (though note that the entries in
(\ref{eq:run-example-encoder-partial}) and their combinations already give
$2^{8}=256$~edges that should be part of the state diagram---we would need the
full specification of the encoder for our running example in order to
determine its state diagram). Figure~8 of Ref.~\cite{PTO09} depicts a simple
example of an encoder that acts on one memory qubit, one ancilla qubit, and
one information qubit. Thus, its state diagram has only four vertices and
32~edges, and Figure~9 of the same paper depicts the encoder's state diagram.

\label{sec:catastrophic-review}We now review the definition of catastrophicity
from Ref.~\cite{PTO09},\footnote{We should note that there have been previous
(flawed) definitions of catastrophicity in the quantum convolutional coding
literature. The first appearing in Ref.~\cite{arxiv2004olliv}\ is erroneous by
the argument in Ref.~\cite{G10}. Suppose that a convolutional encoder cyclicly
permutes the qubits in a frame upward so that the first qubit becomes the
last, and suppose it then follows with a block encoding on the other qubits.
This encoder cannot be arranged into the \textquotedblleft
pearl-necklace\textquotedblright\ form required by Proposition~4.1 of
Ref.~\cite{arxiv2004olliv}, but it nevertheless is obviously non-catastrophic
because errors never propagate between logical qubits in different frames.
\par
The definition of non-catastrophicity in Ref.~\cite{GR06b}\ is also erroneous.
It states that an encoder is non-catastrophic if it can be arranged into a
circuit of finite depth. This definition excludes the class of recursive
quantum convolutional encoders, which cannot be arranged into a circuit of
finite depth. Now, it turns out from a detailed analysis that every recursive
quantum convolutional encoder is catastrophic according to the definition in
Definition~\ref{def:catastrophic} (Theorem~1 of Ref.~\cite{PTO09}), but this
theorem does not apply to entanglement-assisted quantum convolutional encoders
that can be both recursive and non-catastrophic~\cite{WH10b}. Thus, in light
of these latter developments, the definition of non-catastrophicity from
Ref.~\cite{GR06b}\ is flawed.} which is based on the classical notion of
catastrophicity from Ref.~\cite{V71,McE02}. The essential idea behind
catastrophic error propagation is that an error with finite weight, after
being fed through the inverse of the encoder, could propagate infinitely
throughout the decoded information qubit stream without triggering syndromes
corresponding to these errors. The only way that this catastrophic error
propagation can occur is if there is some cycle in the state diagram where all
of the edges along the cycle have physical labels equal to the identity
operator, while at least one of the edges has a logical label that is not
equal to the identity. If such a cycle exists, it implies that the
finite-weight channel error produces an infinite-weight information qubit
error without triggering syndrome bits corresponding to this error (if it did
trigger syndrome bits, this cycle would not be in the state diagram), and an
iterative decoding algorithm such as that presented in Ref.~\cite{PTO09}\ is
not able to detect these errors. So, we can now state the definition of a
catastrophic encoder.

\begin{definition}
[Catastrophic Encoder]\label{def:catastrophic}A quantum convolutional encoder
acting on memory qubits, information qubits, and ancilla qubits is
catastrophic if there exists a cycle in its state diagram where all edges in
the cycle have zero physical weight, but there is at least one edge in the
cycle with non-zero logical weight.\footnote{Interestingly, catastrophicity in
the quantum world is not only a property of the encoder, but it also depends
on the resources on which the encoder acts~\cite{WH10b}. For example, we can
replace the ancilla qubit of the catastrophic encoder in Figure~8 of
Ref.~\cite{PTO09} with one system of an entangled bit, and the resulting
encoder becomes non-catastrophic. This type of thing can never happen
classically if the only kind of resource employed is a classical bit.}
\end{definition}

\subsection{Towards a Minimal-memory/Non-catastrophic Encoder}

\label{sec:main-theorem}This section presents our main results that apply to
the task of finding a minimal-memory, non-catastrophic encoder for an
arbitrary set of stabilizer generators that form a valid quantum convolutional
code. Our first theorem states a sufficient condition for a minimal-memory
encoder to be non-catastrophic, and this theorem applies to our running
example in (\ref{eq:running-example}) and
(\ref{eq:run-example-encoder-partial}).

\subsubsection{Encoders with a Full-rank Memory Commutativity Matrix}

\label{sec:main-theorem1}

\begin{theorem}
\label{thm:min-cat} Suppose the memory commutativity matrix of a given set of
stabilizer generators is full rank. Then any minimal-memory encoder with a
partial specification given by Theorem~\ref{thm:min-mem}\ is non-catastrophic.
\end{theorem}

\begin{proof}
We need to consider an encoder of the general form in
(\ref{eq:general-encoder}). Suppose for a contradiction that the
minimal-memory encoder with $m$ memory qubits is catastrophic. By
Definition~\ref{def:catastrophic}, this implies that there is some cycle
through a set of memory states $\left\{  m_{1},\ldots,m_{p}\right\}  $ of the
following form (with zero physical weight but non-zero logical weight):%
\begin{equation}%
\begin{tabular}
[c]{c|c|c}%
Mem. & Anc. & Info.\\\hline\hline
$m_{1}$ & $s_{1}$ & $l_{1}$\\
$m_{2}$ & $s_{2}$ & $l_{2}$\\
$\vdots$ & $\vdots$ & $\vdots$\\
$m_{p}$ & $s_{p}$ & $l_{p}$%
\end{tabular}
\rightarrow%
\begin{tabular}
[c]{c|c}%
Phys. & Mem.\\\hline\hline
$I^{\otimes n}$ & $m_{2}$\\
$I^{\otimes n}$ & $m_{3}$\\
$\vdots$ & $\vdots$\\
$I^{\otimes n}$ & $m_{1}$%
\end{tabular}
, \label{catastrophic1}%
\end{equation}
where $m_{1}$, \ldots, $m_{p}$ are arbitrary Pauli operators acting on the
memory qubits, the operators $s_{i}\in\{I,Z\}^{\otimes(n-k)}$ act on the $n-k$
ancilla qubits, and the operators $l_{i}$ are arbitrary $k$-qubit Pauli
operators acting on the $k$ information qubits (with at least one of them not
equal to the identity operator). Observe that all of the output rows on the
RHS of~(\ref{catastrophic1}) commute with the last row in each of the $n-k$
blocks on the RHS of the transformation in~(\ref{eq:general-encoder}). This
observation implies that all of the rows on the LHS of~(\ref{catastrophic1})
should commute with the corresponding rows on the LHS\ of the transformation
in~(\ref{eq:general-encoder}). Therefore, all operators $m_{1},m_{2}%
,...,m_{p}$ acting on the memory qubits commute with the memory operators
$g_{i,l_{i}-1}$ for all $i\in\left\{  1,2,...,n-k\right\}  $. Continuing, we
now know that all of the rows on the RHS\ of (\ref{catastrophic1}) commute
with the second-to-last row in each of the $n-k$ blocks on the RHS\ of the
transformation in (\ref{eq:run-example-encoder-partial}). This then implies
that $m_{1},...,m_{p}$ commute with $g_{i,l_{i}-2}$ for all $i\in\left\{
1,2,...,n-k\right\}  $ by the same reasoning above. Continuing in this manner
up the rows of each of the $n-k$ blocks, we can show that the operators
$m_{1},m_{2},...,m_{p}$ commute with all of the memory operators $g_{i,j}$ for
all $i\in\{1,2,...,n-k\}$ and $j\in\{1,2,...,l_{i}-1\}$.

All of these commutativity constraints restrict the form of the operators
$m_{1},...,m_{p}$ in the catastrophic cycle. By assumption, the rank of the
memory commutativity matrix is full and equal to $2m$. This implies that there
are $2m$ memory operators $g_{i,j}$ and they form a complete basis for the
Pauli group on $m$ qubits. It follows that each of the operators
$m_{1},...,m_{p}$ is equal to the identity operator on $m$ qubits because they
are required to commute with all $g_{i,j}$ and the only operator that can do
so is the $m$-qubit identity operator. So all of the entries in
(\ref{catastrophic1}) are really just cycles of the form%
\[%
\begin{tabular}
[c]{c|c|c}%
Mem. & Anc. & Info.\\\hline\hline
$I^{\otimes m}$ & $s_{1}$ & $l_{1}$\\
$I^{\otimes m}$ & $s_{2}$ & $l_{2}$\\
$\vdots$ & $\vdots$ & $\vdots$\\
$I^{\otimes m}$ & $s_{p}$ & $l_{p}$%
\end{tabular}
\rightarrow%
\begin{tabular}
[c]{c|c}%
Phys. & Mem.\\\hline\hline
$I^{\otimes n}$ & $I^{\otimes m}$\\
$I^{\otimes n}$ & $I^{\otimes m}$\\
$\vdots$ & $\vdots$\\
$I^{\otimes n}$ & $I^{\otimes m}$%
\end{tabular}
.
\]
The above input-output relations restrict $s_{1},...,s_{p}$, and $l_{1}%
,\ldots,l_{p}$ further---it is impossible for $s_{1},...,s_{p}$, and
$l_{1},\ldots,l_{p}$ to be any Pauli operator besides the identity operator.
Otherwise, the encoder would not transform the entry on the LHS to the all
identity operator. Thus, the only cycle of zero-physical weight in a
minimal-memory encoder given by Theorem~\ref{thm:min-mem}\ that implements the
transformation in~(\ref{eq:general-encoder}) is the self-loop at the identity
memory state with zero logical weight, which implies the encoder is non-catastrophic.
\end{proof}

We return to our running example from (\ref{eq:running-example}). We
determined in (\ref{eq:run-example-encoder-partial}) a partial specification
of a minimal-memory encoder for these generators, and the above theorem states
that any encoder that realizes this transformation is non-catastrophic as
well. Indeed, we can study the proof technique above for this example. Suppose
for a contradiction that a catastrophic cycle exists in the state diagram for
the minimal-memory encoder in (\ref{eq:run-example-encoder-partial}). Such a
catastrophic cycle has the following form:%
\begin{equation}%
\begin{tabular}
[c]{c|cc|cc}%
Mem. & \multicolumn{2}{|c}{Anc.} & \multicolumn{2}{|c}{Info.}\\\hline\hline
$m_{1}$ & $s_{1,1}$ & $s_{1,2}$ & $l_{1,1}$ & $l_{1,2}$\\
$m_{2}$ & $s_{2,1}$ & $s_{2,2}$ & $l_{2,1}$ & $l_{2,2}$\\
$\vdots$ & $\vdots$ & $\vdots$ & $\vdots$ & $\vdots$\\
$m_{p}$ & $s_{p,1}$ & $s_{p,2}$ & $l_{p,1}$ & $l_{p,2}$%
\end{tabular}
\rightarrow%
\begin{tabular}
[c]{c|c}%
\multicolumn{1}{c|}{Phys.} & Mem.\\\hline\hline
$I^{\otimes4}$ & $m_{2}$\\
$I^{\otimes4}$ & $m_{3}$\\
$\vdots$ & $\vdots$\\
$I^{\otimes4}$ & $m_{1}$%
\end{tabular}
\label{eq:catastrophic-cycle}%
\end{equation}
where $m_{1}$, \ldots, $m_{p}$ can be arbitrary Pauli operators acting on the
three memory qubits, each $s_{i,j}\in\{I,Z\}$ acts on an ancilla qubit, and
each $l_{i,j}$ is an arbitrary single-qubit Pauli operator acting on an
information qubit (with at least one $l_{i,j}$ not equal to the identity
operator). Observe that all of the output rows on the RHS
of~(\ref{eq:catastrophic-cycle}) commute with the fourth and eighth rows on
the RHS of the transformation in~(\ref{eq:run-example-encoder-partial}). This
observation implies that all of the rows on the LHS
of~(\ref{eq:catastrophic-cycle}) should commute with the fourth and eighth
rows on the LHS\ of the transformation
in~(\ref{eq:run-example-encoder-partial}). Therefore all operators
$m_{1},m_{2},...,m_{p}$ acting on the memory qubits commute with $IZI$ and
$ZII$. Continuing, we now know that all of the rows on the RHS\ of
(\ref{eq:catastrophic-cycle}) commute with the third and seventh rows of
(\ref{eq:run-example-encoder-partial}) because $\left[  m_{i},Z_{2}\right]
=\left[  m_{i},Z_{1}\right]  =0$ for all $1\leq i\leq p$. This then implies
that $m_{1},...,m_{p}$ commute with $IIX$ and $IIZ$ by the same reasoning
above. We can continue one last time to show that all $m_{1},...,m_{p}$
commute with $XIX$ and $ZXZ$. Similar to the reasoning in the above theorem,
all of these commutativity constraints restrict the form of the operators
$m_{1},...,m_{p}$ in the catastrophic cycle. In fact, the only three-qubit
operator that commutes with $IZI$, $ZII$, $IIZ$, $IIZ$, $XIX$, and $ZXZ$ is
the three-qubit identity operator because the aforementioned operators form a
complete basis for the Pauli group on three qubits. Applying the same logic as
at the end of the above proof then allows us to conclude that the encoder is non-catastrophic.

\subsubsection{Encoders without a Full-rank Memory Commutativity Matrix and
with an Empty Partial Null Space}

Now suppose that the memory commutativity matrix of a given set of stabilizer
generators is not full rank. As we explained in the proof of
Theorem~\ref{thm:min-cat}, the memory operators $m_{1},\cdots,m_{p}$ of a
catastrophic cycle in~(\ref{catastrophic1}) commute with all memory operators
$g_{i,j}$ in~(\ref{eq:general-encoder}). Since the number of commutativity
constraints is less than $2m$ in this case (where $m$ is the number of qubits
on which the memory operators act), there are other choices for the
catastrophic memory operators $m_{1},\cdots,m_{p}$ besides the $m$-qubit
identity operator that are consistent with these constraints. This implies
that some of the encoders implementing the transformation
in~(\ref{eq:general-encoder}) may be catastrophic. To illustrate this case, we
choose the second code of Figure~1 in Ref.~\cite{GR07} as another running
example. This code has the following stabilizer generators:%
\[%
\begin{array}
[c]{c}%
h_{1}=XXXX\\
h_{2}=ZZZZ
\end{array}
\left\vert
\begin{array}
[c]{c}%
XXII\\
ZZII
\end{array}
\right\vert
\begin{array}
[c]{c}%
IXIX\\
IZIZ
\end{array}
\left\vert
\begin{array}
[c]{c}%
IIXX\\
IIZZ
\end{array}
\right\vert
\begin{array}
[c]{c}%
XXXX\\
ZZZZ
\end{array}
.
\]
An encoding unitary for this code should be as follows:%
\[%
\begin{tabular}
[c]{c|cc|cc}%
Mem. & \multicolumn{2}{|c|}{Anc.} & \multicolumn{2}{|c}{Info.}\\\hline\hline
$I^{\otimes m}$ & $Z$ & $I$ & $I$ & $I$\\
$g_{1,1}$ & $I$ & $I$ & $I$ & $I$\\
$g_{1,2}$ & $I$ & $I$ & $I$ & $I$\\
$g_{1,3}$ & $I$ & $I$ & $I$ & $I$\\
$g_{1,4}$ & $I$ & $I$ & $I$ & $I$\\\hline
$I^{\otimes m}$ & $I$ & $Z$ & $I$ & $I$\\
$g_{2,1}$ & $I$ & $I$ & $I$ & $I$\\
$g_{2,2}$ & $I$ & $I$ & $I$ & $I$\\
$g_{2,3}$ & $I$ & $I$ & $I$ & $I$\\
$g_{2,4}$ & $I$ & $I$ & $I$ & $I$%
\end{tabular}
\ \rightarrow%
\begin{tabular}
[c]{cccc|c}%
\multicolumn{4}{c|}{Phys.} & Mem.\\\hline\hline
$X$ & $X$ & $X$ & $X$ & $g_{1,1}$\\
$X$ & $X$ & $I$ & $I$ & $g_{1,2}$\\
$I$ & $X$ & $I$ & $X$ & $g_{1,3}$\\
$I$ & $I$ & $X$ & $X$ & $g_{1,4}$\\
$X$ & $X$ & $X$ & $X$ & $I^{\otimes m}$\\\hline
$Z$ & $Z$ & $Z$ & $Z$ & $g_{2,1}$\\
$Z$ & $Z$ & $I$ & $I$ & $g_{2,2}$\\
$I$ & $Z$ & $I$ & $Z$ & $g_{2,3}$\\
$I$ & $I$ & $Z$ & $Z$ & $g_{2,4}$\\
$Z$ & $Z$ & $Z$ & $Z$ & $I^{\otimes m}$%
\end{tabular}
\ .
\]
By inspecting the commutativity relations of the memory operators~$g_{i,j}$ in
the above transformation, the commutativity matrix is%
\[
\Omega=%
\begin{bmatrix}
0 & 0 & 0 & 0 & 0 & 0 & 0 & 0\\
0 & 0 & 0 & 0 & 0 & 0 & 1 & 0\\
0 & 0 & 0 & 0 & 0 & 1 & 0 & 0\\
0 & 0 & 0 & 0 & 0 & 0 & 0 & 0\\
0 & 0 & 0 & 0 & 0 & 0 & 0 & 0\\
0 & 0 & 1 & 0 & 0 & 0 & 0 & 0\\
0 & 1 & 0 & 0 & 0 & 0 & 0 & 0\\
0 & 0 & 0 & 0 & 0 & 0 & 0 & 0
\end{bmatrix}
,
\]
with dimension equal to eight and rank equal to four. So, based on
Theorem~\ref{thm:min-mem}, the minimal number of required memory qubits is
six. A set of memory operators which act on a minimal number of qubits is as
follows:%
\begin{align*}
g_{1,1}  &  =ZIIIII,\ \ \ g_{1,2}=IIXIII,\ \ \ g_{1,3}=IIIZII,\\
g_{1,4}  &  =IIIIZI,\ \ \ g_{2,1}=IZIIII,\ \ \ g_{2,2}=IIIXII,\\
g_{2,3}  &  =IIZIII,\ \ \ g_{2,4}=IIIIIZ.
\end{align*}
Thus, the encoder implements the following transformation:%
\begin{equation}%
\begin{tabular}
[c]{c|c|c}%
Mem. & \multicolumn{1}{|c|}{Anc.} & Info.\\\hline\hline
$IIIIII$ & $ZI$ & $II$\\
$ZIIIII$ & $II$ & $II$\\
$IIXIII$ & $II$ & $II$\\
$IIIZII$ & $II$ & $II$\\
$IIIIZI$ & $II$ & $II$\\\hline
$IIIIII$ & $IZ$ & $II$\\
$IZIIII$ & $II$ & $II$\\
$IIIXII$ & $II$ & $II$\\
$IIZIII$ & $II$ & $II$\\
$IIIIIZ$ & $II$ & $II$%
\end{tabular}
\rightarrow%
\begin{tabular}
[c]{c|c}%
\multicolumn{1}{c|}{Phys.} & Mem.\\\hline\hline
$XXXX$ & $ZIIIII$\\
$XXII$ & $IIXIII$\\
$IXIX$ & $IIIZII$\\
$IIXX$ & $IIIIZI$\\
$XXXX$ & $IIIIII$\\\hline
$ZZZZ$ & $IZIIII$\\
$ZZII$ & $IIIXII$\\
$IZIZ$ & $IIZIII$\\
$IIZZ$ & $IIIIIZ$\\
$ZZZZ$ & $IIIIII$%
\end{tabular}
. \label{eq:second-running-ex-trans}%
\end{equation}
When the commutativity matrix is not full rank, we should add some rows to the
transformation in~(\ref{eq:general-encoder}) to ensure that the encoder
implementing the transformation is non-catastrophic. To fulfill this
requirement, the first step is to find a set$~C$ of memory states that can be
a part of catastrophic cycle (i.e., memory states which satisfy the
commutativity relations mentioned in the proof of Theorem~\ref{thm:min-cat}).
In our running example, the memory operators in a catastrophic cycle should
commute with $ZIIIII$, $IIXIII$, $IIIZII$, $IIIIZI$, $IZIIII$, $IIIXII$,
$IIZIII$ and $IIIIIZ$. Thus, they must be an operator in the following set:%
\[
C=\{Z_{1}^{e_{1}}Z_{2}^{e_{2}}Z_{5}^{e_{3}}Z_{6}^{e_{4}}:e_{1},e_{2}%
,e_{3},e_{4}\in\{0,1\}\}\label{mem-cat}.
\]
The next step is to search among the rows and their combinations
in~(\ref{eq:general-encoder}) to find a set $S_{1}$ whose members can
potentially be a part of catastrophic cycle. Entries in $S_{1}$ have the
following form:%

\[%
\begin{tabular}
[c]{c|c|c}%
Mem. & Anc. & Info.\\\hline\hline
$M$ & $S^{z}$ & $L$%
\end{tabular}
\rightarrow%
\begin{tabular}
[c]{c|c}%
Phys. & Mem.\\\hline\hline
$I^{\otimes{n}}$ & $M^{\prime}$%
\end{tabular}
,
\]
where $M$ and $M^{\prime}$ are both elements of the set $C$, the operator
$S^{z}\in\{I,Z\}^{\otimes(n-k)}$ acts on $n-k$ ancilla qubits, and $L$ is an
arbitrary $k$-qubit Pauli operator acting on the information qubits.

In our running example in (\ref{eq:second-running-ex-trans}), members of the
set $S_{1}$ are obtained by adding the first row to the fifth row, and the
sixth row to the tenth row, giving the following input-output relations:%
\begin{equation}%
\begin{tabular}
[c]{c|c|c}%
Mem. & \multicolumn{1}{|c|}{Anc.} & Info.\\\hline\hline
$IIIIZI$ & $ZI$ & $II$\\
$IIIIIZ$ & $IZ$ & $II$%
\end{tabular}
\rightarrow%
\begin{tabular}
[c]{c|c}%
\multicolumn{1}{c|}{Phys.} & Mem.\\\hline\hline
$IIII$ & $ZIIIII$\\
$IIII$ & $IZIIII$%
\end{tabular}
. \label{eq:runn-ex-partial-null}%
\end{equation}
Consider a set$~S_{2}$ of rows with physical output equal to the four-qubit
identity operator. Now we should add such a set of rows to the transformation
so that the output memory operators of the members of $S_{1}$ and $S_{2}$ make
a complete basis for the set $C$. This guarantees that the rows potentially
part of catastrophic cycle are just an entry or a combination of entries of
$S_{1}\cup S_{2}$. So if we choose the memory states of the elements of
$S_{2}$ such that the set $S_{1}\cup S_{2}$ does not create a catastrophic
cycle, we can ensure that any encoders performing the transformation with the
added rows will be non-catastrophic. In our running example from
(\ref{eq:second-running-ex-trans}) and (\ref{eq:runn-ex-partial-null}), we
just add two new rows (the rows after the line) as follows:%

\begin{equation}%
\begin{tabular}
[c]{c|c|c}%
Mem. & \multicolumn{1}{|c|}{Anc.} & Info.\\\hline\hline
$IIIIZI$ & $ZI$ & $II$\\
$IIIIIZ$ & $IZ$ & $II$\\\hline
$IIIIII$ & $II$ & $XI$\\
$IIIIII$ & $II$ & $IX$%
\end{tabular}
\rightarrow%
\begin{tabular}
[c]{c|c}%
\multicolumn{1}{c|}{Phys.} & Mem.\\\hline\hline
$IIII$ & $ZIIIII$\\
$IIII$ & $IZIIII$\\\hline
$IIII$ & $IIIIZI$\\
$IIII$ & $IIIIIZ$%
\end{tabular}
\ . \label{eq:non-cat-added-row}%
\end{equation}
All combinations of the entries in~(\ref{eq:non-cat-added-row}) are as
follows:
\begin{equation}%
\begin{tabular}
[c]{c|c|c}%
Mem. & \multicolumn{1}{|c|}{Anc.} & Info.\\\hline\hline
$IIIIZI$ & $ZI$ & $II$\\
$IIIIIZ$ & $IZ$ & $II$\\
$IIIIII$ & $II$ & $XI$\\
$IIIIII$ & $II$ & $IX$\\
$IIIIZZ$ & $ZZ$ & $II$\\
$IIIIZI$ & $ZI$ & $XI$\\
$IIIIZI$ & $ZI$ & $IX$\\
$IIIIIZ$ & $IZ$ & $XI$\\
$IIIIIZ$ & $IZ$ & $IX$\\
$IIIIII$ & $II$ & $XX$\\
$IIIIZZ$ & $ZZ$ & $XI$\\
$IIIIZZ$ & $ZZ$ & $IX$\\
$IIIIZI$ & $ZI$ & $XX$\\
$IIIIIZ$ & $IZ$ & $XX$\\
$IIIIZZ$ & $ZZ$ & $XX$%
\end{tabular}
\rightarrow%
\begin{tabular}
[c]{c|c}%
\multicolumn{1}{c|}{Phys.} & Mem.\\\hline\hline
$IIII$ & $ZIIIII$\\
$IIII$ & $IZIIII$\\
$IIII$ & $IIIIZI$\\
$IIII$ & $IIIIIZ$\\
$IIII$ & $ZZIIII$\\
$IIII$ & $ZIIIZI$\\
$IIII$ & $ZIIIIZ$\\
$IIII$ & $IZIIZI$\\
$IIII$ & $IZIIIZ$\\
$IIII$ & $IIIIZZ$\\
$IIII$ & $ZZIIZI$\\
$IIII$ & $ZZIIIZ$\\
$IIII$ & $ZIIIZZ$\\
$IIII$ & $IZIIZZ$\\
$IIII$ & $ZZIIZZ$%
\end{tabular}
. \label{eq:non-cat-added-row-comb}%
\end{equation}
By inspecting the rows in~(\ref{eq:non-cat-added-row-comb}) it is clear that
there is no catastrophic cycle.

Theorem~\ref{thm:empty-cat} below generalizes the technique from the above
example to give a straightforward way for adding rows when $S_{1}$ is an empty set.

\begin{theorem}
\label{thm:empty-cat} Suppose the memory commutativity matrix of a given set
of stabilizer generators is not full rank, and suppose further that the set
$S_{1}$ corresponding to the transformation is an empty set. Then adding
rows in the following form to the transformation in~(\ref{eq:general-encoder})
ensures that any encoder implementing the transformation is non-catastrophic:%
\begin{equation}%
\begin{tabular}
[c]{c|c|c}%
Mem. & Anc. & Info.\\\hline\hline
$I^{\otimes m}$ & $I^{\otimes(n-k)}$ & $X_{1}$\\
$\vdots$ & $\vdots$ & $\vdots$\\
$I^{\otimes m}$ & $I^{\otimes(n-k)}$ & $X_{a}$%
\end{tabular}
\rightarrow%
\begin{tabular}
[c]{c|c}%
Phys. & Mem.\\\hline\hline
$I^{\otimes n}$ & $M_{1}$\\
$\vdots$ & $\vdots$\\
$I^{\otimes n}$ & $M_{a}$%
\end{tabular}
, \label{add-row}%
\end{equation}
where $X_{i}$ denotes the Pauli $X$ operator acting on the $i^{\text{th}}$
information qubit and the operators $M_{1},\ldots,M_{a}$ form a complete basis
for the set $C$.
\end{theorem}

\begin{proof}
Suppose for a contradiction that the entries in~(\ref{add-row}) create a
catastrophic cycle. Since all input memory operators in~(\ref{add-row}) are
equal to the $m$-qubit identity operator, the output memory operator of the
last row in a catastrophic cycle in~(\ref{catastrophic1}) should be equal to the
identity as well (so that the sequence of memory states forms a cycle). This
implies that the last row of the catastrophic cycle is as follows:%
\[%
\begin{tabular}
[c]{c|c|c}%
Mem. & Anc. & Info.\\\hline\hline
$m_{p}$ & $S^{z}$ & $L$%
\end{tabular}
\rightarrow%
\begin{tabular}
[c]{c|c}%
Phys. & Mem.\\\hline\hline
$I^{\otimes{n}}$ & $I^{\otimes k}$%
\end{tabular}
.
\]
So $m_{p}$ and consequently $m_{p-1},...,m_{1}$ are all equal to the $m$-qubit
identity operator. Thus all of the entries in~(\ref{catastrophic1}) are really
just cycles of the following form:%
\[%
\begin{tabular}
[c]{c|c|c}%
Mem. & Anc. & Info.\\\hline\hline
$I^{\otimes m}$ & $s_{1}$ & $l_{1}$\\
$I^{\otimes m}$ & $s_{2}$ & $l_{2}$\\
$\vdots$ & $\vdots$ & $\vdots$\\
$I^{\otimes m}$ & $s_{p}$ & $l_{p}$%
\end{tabular}
\rightarrow%
\begin{tabular}
[c]{c|c}%
Phys. & Mem.\\\hline\hline
$I^{\otimes n}$ & $I^{\otimes m}$\\
$I^{\otimes n}$ & $I^{\otimes m}$\\
$\vdots$ & $\vdots$\\
$I^{\otimes n}$ & $I^{\otimes m}$%
\end{tabular}
.
\]
The above input-output relations imply that $s_{1},...,s_{p}$ and
$l_{1},\ldots,l_{p}$ are identity operators (otherwise, it would not be
possible to effect the above transformation). Thus, the only cycle of
zero-physical weight is the self-loop at the identity memory state with zero
logical weight, which implies there is no catastrophic cycle.
\end{proof}

\section{Non-recursiveness}

\label{sec:non-rec}

In this section, we demonstrate that the encoders from both Theorems~\ref{thm:min-cat}
and \ref{thm:empty-cat} are non-recursive. Recursiveness or lack thereof is a fundamental
property of a quantum convolutional encoder as demonstrated in Ref.~\cite{PTO09}. In
Ref.~\cite{PTO09}, Poulin \textit{et al}.~proved that any non-catastrophic quantum
convolutional encoder is already non-recursive. Note that this situation is
much different from classical convolutional encoders for which these two
properties are not directly linked. In light of the results of Poulin
\textit{et al}, it follows that our encoders from Theorems~\ref{thm:min-cat} and \ref{thm:empty-cat} are
non-recursive because they are already non-catastrophic. Nevertheless, we
prove below that the encoders are non-recursive because our proof technique is arguably
much simpler than the proof of Theorem~1 from Ref.~\cite{PTO09}. Though, before proving
these theorems, we briefly review the definition of recursiveness.

\begin{definition}
[Recursive encoder]\label{def:recursive}An admissable path is a path in the
state diagram for which its first edge is not part of a zero physical-weight
cycle. Consider any vertex belonging to a zero physical-weight loop and any
admissable path beginning at this vertex that also has logical weight one. The
encoder is recursive if all such paths do not contain a zero physical-weight loop.
\end{definition}

We can gain some intuition behind the above definition by recalling the
definition of a recursive classical convolutional encoder. In the classical
case, an encoder is recursive if it has an infinite impulse response---that
is, if it outputs an infinite-weight, periodic sequence in response to an
input consisting of a single \textquotedblleft one\textquotedblright\ followed
by an infinite number of \textquotedblleft zeros.\textquotedblright\
  Definition~\ref{def:recursive} above for the quantum case ensures that the
response to a single Pauli operator (one of $\{X,Y,Z\}$) at a single logical
input along with the identity operator at all other logical inputs leads to a
periodic output sequence of Pauli operators with infinite weight. Though, the
definition above ensures that this is not only the case for the above sequence
but also for one in which the ancilla qubit inputs can be chosen arbitrarily
from $\left\{  I,Z\right\}  $. Thus, it is a much more stringent condition for
a quantum convolutional encoder to be recursive.

We are now in a position to prove the main theorem of this section.

\begin{theorem}
The encoders from Theorems~\ref{thm:min-cat} and \ref{thm:empty-cat} are non-recursive in addition to being non-catastrophic.
\end{theorem}

\begin{IEEEproof}
In order to prove that an encoder is non-recursive, we just need to find a
single logical-weight-one admissable path beginning and ending in the identity
memory state.
First consider that every memory state in (\ref{eq:general-encoder}) already has a zero-logical-weight path back to the
identity memory state. (For example, for the entry $g_{1,1}$ in the second
row, one would just need to input $I^{\otimes k}$ and $I^{\otimes n-k}$ for the logical inputs and
ancillas, which in turn leads to state $g_{1,2}$. Continuing in this fashion
leads to the state $g_{1,l_{1}-1}$, which finally leads to the identity memory
state.) 

Now consider the encoders from Theorem~\ref{thm:min-cat} and
consider further the following transformation:
\[I^{\otimes m}\otimes I^{\otimes n-k}\otimes X_{i}\rightarrow h \otimes g,\]
where $h$ is some arbitrary $n$-qubit Pauli operator and $g$ is some $m$-qubit Pauli operator.
From the fact that the memory commutativity matrix is full rank,
we know that it is possible to construct the memory state $g$\ by combining
the memory states from (\ref{eq:general-encoder}) (say, for example, $g=g_{i_{1},j_{1}}\cdot
g_{i_{2},j_{2}}\cdot\cdots\cdot g_{i_{m},j_{m}}$). Furthermore, by inputting
$I^{\otimes k}$ and $I^{\otimes n-k}$ for all subsequent logical and ancilla inputs, we can
construct a path that is a combination of the paths taken by each of
$g_{i_{1},j_{1}}$, $g_{i_{2},j_{2}}$, \ldots, $g_{i_{m},j_{m}}$. Since all of
these paths end up in the identity memory state, it follows that the
combination of the paths also ends up in the identity memory state.
So there is a logical-weight-one admissable path
beginning and ending in the identity memory state.
This concludes the proof for encoders from Theorem~\ref{thm:min-cat}.

The proof for the encoders from Theorem~\ref{thm:empty-cat} is similar to the above proof.
First, let us consider the
memory states that are part of the set $C$. The rows in (\ref{add-row})\ added to the
transformation are all weight-one logical edges from the identity memory state
to a state in $C$ because they have the following form:
 \[I^{\otimes m}\otimes I^{\otimes n-k}\otimes X_{i}\rightarrow I^{\otimes n}\otimes M_{i}.\]
Since all of the
memory states in (\ref{eq:general-encoder}) commute with the elements of $C$, we can combine some of these
commuting states together to realize the memory state $M_{i}$. By the same
argument as before, inputting $I^{\otimes k}$ and $I^{\otimes n-k}$ for all subsequent logical and
ancilla inputs eventually leads back to the identity memory state because all
of the individual paths lead back as seen in (\ref{eq:general-encoder}). This concludes the proof for
encoders from Theorem~\ref{thm:empty-cat}.
\end{IEEEproof}

\section{Conclusion}

\label{sec:conclusion}We have presented an algorithm to find a
minimal-memory, non-catastrophic, polynomial-depth encoder for a given set of
stabilizer generators. Our algorithm first determines a transformation that
the encoder should perform, without specifying the Pauli operators acting on
the memory qubits. It then finds a set of Pauli operators which act on a
minimal number of memory qubits and are consistent with the input-output
commutation relations of the encoder. The number of minimal memory qubits
depends on the dimension and the rank of the \textquotedblleft memory
commutativity matrix,\textquotedblright\ which details the commutativity
relations between the memory operators. Once the memory operators are
determined, there is a polynomial-time algorithm to find the encoder which
performs the transformation. We have also proved that any minimal-memory
encoder with a full-rank memory commutativity matrix is non-catastrophic.
However, when the memory commutativity matrix is not full-rank, we should add
some rows to the transformation to ensure that the encoder is
non-catastrophic. Theorem~\ref{thm:empty-cat} includes an explicit way of
adding rows to transformations that have an empty partial null space. We proved
that the encoders from Theorems~\ref{thm:min-cat} and \ref{thm:empty-cat} are non-recursive in addition to being non-catastrophic. Finally, the appendix contains details of our algorithm for many examples
of quantum convolutional codes from Refs.~\cite{ieee2007forney,GR07}.

Some open questions still remain.
First, we are assuming a particular form for our encoders, that they have to take
the unencoded Pauli $Z$ operators to the encoded stabilizer operators.
Although this form for the encoder is natural,
it might be the case that allowing for a different form could lead to encoders
with smaller memory requirements.
Another open problem is to find an explicit way of adding rows to any
transformation without a full-rank memory commutativity matrix in order to
ensure that the encoder is non-catastrophic. It is also an open problem to
find minimal-memory, non-catastrophic encoders for subsystem convolutional
codes~\cite{poulin:230504,WB08a}, entanglement-assisted quantum convolutional
codes~\cite{WB07}, and convolutional codes that send both classical and
quantum information~\cite{WB08a}.

\section*{Acknowledgements}

The authors are grateful to Markus Grassl, Johannes G\"{u}tschow, David
Poulin, and Martin R\"{o}tteler for useful discussions. MH and SHK acknowledge
support from the Iranian Telecommunication Research Center (ITRC). MMW
acknowledges support from the MDEIE (Quebec) PSR-SIIRI international
collaboration grant.

\bibliographystyle{IEEEtran}
\bibliography{Ref}
\onecolumn
\pagebreak

\appendix

\section{Appendix}

This section includes many examples of quantum convolutional codes from
Refs.~\cite{ieee2007forney,GR07} with non-full rank memory commutativity
matrix. For all of them we state how to add rows to be confident that the
minimal-memory encoder implementing the transformation is non-catastrophic.
\label{sec:many-examples}

\subsection{First Example}

The second example in Table I in \cite{ieee2007forney} has the following two
stabilizer generators:
\[%
\begin{array}
[c]{cccc}%
X & X & X & X\\
Z & Z & Z & Z
\end{array}
\left\vert
\begin{array}
[c]{cccc}%
I & I & X & X\\
I & I & Z & Z
\end{array}
\right\vert \left.
\begin{array}
[c]{cccc}%
I & X & I & X\\
I & Z & I & Z
\end{array}
\right\vert
\begin{array}
[c]{cccc}%
I & I & X & X\\
I & I & Z & Z
\end{array}
.
\]
So the encoder should act as follows:%
\begin{equation}%
\begin{tabular}
[c]{c|cc|cc}%
Mem. & \multicolumn{2}{|c|}{Anc.} & \multicolumn{2}{|c}{Info.}\\\hline\hline
$I^{\otimes{m}}$ & $Z$ & $I$ & $I$ & $I$\\
$g_{1,1}$ & $I$ & $I$ & $I$ & $I$\\
$g_{1,2}$ & $I$ & $I$ & $I$ & $I$\\
$g_{1,3}$ & $I$ & $I$ & $I$ & $I$\\\hline
$I^{\otimes{m}}$ & $I$ & $Z$ & $I$ & $I$\\
$g_{2,1}$ & $I$ & $I$ & $I$ & $I$\\
$g_{2,2}$ & $I$ & $I$ & $I$ & $I$\\
$g_{2,3}$ & $I$ & $I$ & $I$ & $I$%
\end{tabular}
\ \ \ \ \ \ \ \rightarrow\ \ \ \ \
\begin{tabular}
[c]{cccc|c}%
\multicolumn{4}{c|}{Phys.} & Mem.\\\hline\hline
$X$ & $X$ & $X$ & $X$ & $g_{1,1}$\\
$I$ & $I$ & $X$ & $X$ & $g_{1,2}$\\
$I$ & $X$ & $I$ & $X$ & $g_{1,3}$\\
$I$ & $I$ & $X$ & $X$ & $I^{\otimes{m}}$\\\hline
$Z$ & $Z$ & $Z$ & $Z$ & $g_{2,1}$\\
$I$ & $I$ & $Z$ & $Z$ & $g_{2,2}$\\
$I$ & $Z$ & $I$ & $Z$ & $g_{2,3}$\\
$I$ & $I$ & $Z$ & $Z$ & $I^{\otimes{m}}$%
\end{tabular}
\ .
\end{equation}

The commutativity matrix corresponding the above transformation is:%
\[
\Omega=%
\begin{bmatrix}
0 & 0 & 0 & 0 & 0 & 0\\
0 & 0 & 0 & 0 & 0 & 1\\
0 & 0 & 0 & 0 & 1 & 0\\
0 & 0 & 0 & 0 & 0 & 0\\
0 & 0 & 1 & 0 & 0 & 0\\
0 & 1 & 0 & 0 & 0 & 0
\end{bmatrix}
.
\]
Dimension of $\Omega$ is six and its rank is equal to four, so based on
Theorem~\ref{thm:min-mem} the encoder requires at least four memory qubits. A
set of generators with minimal amount of required memory is:%
\begin{align*}
g_{1,1}  &  =ZIII,\ \ \ \ g_{1,2}=IIZI,\ \ \ \ g_{1,3}=IIIX,\\
g_{2,1}  &  =IZII,\ \ \ \ g_{2,2}=IIIZ,\ \ \ \ g_{2,3}=IIXI.
\end{align*}
Thus, the minimal-memory encoder implements the following transformation:
\[%
\begin{tabular}
[c]{cccc|cc|cc}%
\multicolumn{4}{c}{Mem.} & \multicolumn{2}{|c|}{Anc.} &
\multicolumn{2}{|c}{Info.}\\\hline\hline
$I$ & $I$ & $I$ & $I$ & $Z$ & $I$ & $I$ & $I$\\
$Z$ & $I$ & $I$ & $I$ & $I$ & $I$ & $I$ & $I$\\
$I$ & $I$ & $Z$ & $I$ & $I$ & $I$ & $I$ & $I$\\
$I$ & $I$ & $I$ & $X$ & $I$ & $I$ & $I$ & $I$\\\hline
$I$ & $I$ & $I$ & $I$ & $I$ & $Z$ & $I$ & $I$\\
$I$ & $Z$ & $I$ & $I$ & $I$ & $I$ & $I$ & $I$\\
$I$ & $I$ & $I$ & $Z$ & $I$ & $I$ & $I$ & $I$\\
$I$ & $I$ & $X$ & $I$ & $I$ & $I$ & $I$ & $I$%
\end{tabular}
\ \ \ \ \ \ \ \ \ \rightarrow\ \ \ \ \
\begin{tabular}
[c]{cccc|cccc}%
\multicolumn{4}{c|}{Phys.} & \multicolumn{4}{|c}{Mem.}\\\hline\hline
$X$ & $X$ & $X$ & $X$ & $Z$ & $I$ & $I$ & $I$\\
$I$ & $I$ & $X$ & $X$ & $I$ & $I$ & $Z$ & $I$\\
$I$ & $X$ & $I$ & $X$ & $I$ & $I$ & $I$ & $X$\\
$I$ & $I$ & $X$ & $X$ & $I$ & $I$ & $I$ & $I$\\\hline
$Z$ & $Z$ & $Z$ & $Z$ & $I$ & $Z$ & $I$ & $I$\\
$I$ & $I$ & $Z$ & $Z$ & $I$ & $I$ & $I$ & $Z$\\
$I$ & $Z$ & $I$ & $Z$ & $I$ & $I$ & $X$ & $I$\\
$I$ & $I$ & $Z$ & $Z$ & $I$ & $I$ & $I$ & $I$%
\end{tabular}
\ .
\]
The memory operators which can be a part of catastrophic cycle are the members
of following set:
\[
C=\{Z_{1}^{e_{1}}Z_{2}^{e_{2}}:e_{1},e_{2}\in\{0,1\}\}.
\]

Based on Theorem~\ref{thm:min-mem}, since the set $S_{1}$ corresponding the
above transformation is an empty set, assigning any basis of set $C$ to output
memories of the rows of $S_{2}$ will make the encoder non-catastrophic. So any
encoder implementing the following transformation is non-catastrophic.
\[%
\begin{tabular}
[c]{cccc|cc|cc}%
\multicolumn{4}{c}{Mem.} & \multicolumn{2}{|c|}{Anc.} &
\multicolumn{2}{|c}{Info.}\\\hline\hline
$I$ & $I$ & $I$ & $I$ & $Z$ & $I$ & $I$ & $I$\\
$Z$ & $I$ & $I$ & $I$ & $I$ & $I$ & $I$ & $I$\\
$I$ & $I$ & $Z$ & $I$ & $I$ & $I$ & $I$ & $I$\\
$I$ & $I$ & $I$ & $X$ & $I$ & $I$ & $I$ & $I$\\\hline
$I$ & $I$ & $I$ & $I$ & $I$ & $Z$ & $I$ & $I$\\
$I$ & $Z$ & $I$ & $I$ & $I$ & $I$ & $I$ & $I$\\
$I$ & $I$ & $I$ & $Z$ & $I$ & $I$ & $I$ & $I$\\
$I$ & $I$ & $X$ & $I$ & $I$ & $I$ & $I$ & $I$\\\hline
$I$ & $I$ & $I$ & $I$ & $I$ & $I$ & $X$ & $I$\\
$I$ & $I$ & $I$ & $I$ & $I$ & $I$ & $I$ & $X$%
\end{tabular}
\ \ \ \ \ \ \ \ \ \rightarrow\ \ \ \ \
\begin{tabular}
[c]{cccc|cccc}%
\multicolumn{4}{c|}{Phys.} & \multicolumn{4}{|c}{Mem.}\\\hline\hline
$X$ & $X$ & $X$ & $X$ & $Z$ & $I$ & $I$ & $I$\\
$I$ & $I$ & $X$ & $X$ & $I$ & $I$ & $Z$ & $I$\\
$I$ & $X$ & $I$ & $X$ & $I$ & $I$ & $I$ & $X$\\
$I$ & $I$ & $X$ & $X$ & $I$ & $I$ & $I$ & $I$\\\hline
$Z$ & $Z$ & $Z$ & $Z$ & $I$ & $Z$ & $I$ & $I$\\
$I$ & $I$ & $Z$ & $Z$ & $I$ & $I$ & $I$ & $Z$\\
$I$ & $Z$ & $I$ & $Z$ & $I$ & $I$ & $X$ & $I$\\
$I$ & $I$ & $Z$ & $Z$ & $I$ & $I$ & $I$ & $I$\\\hline
$I$ & $I$ & $I$ & $I$ & $Z$ & $I$ & $I$ & $I$\\
$I$ & $I$ & $I$ & $I$ & $I$ & $Z$ & $I$ & $I$%
\end{tabular}
\ .
\]

\subsection{Second Example}

The third example in Table I in~\cite{ieee2007forney} has two following generators:%

\[%
\begin{array}
[c]{cccc}%
X & X & X & X\\
Z & Z & Z & Z
\end{array}
\left\vert
\begin{array}
[c]{cccc}%
X & X & I & I\\
Z & Z & I & I
\end{array}
\right\vert \left.
\begin{array}
[c]{cccc}%
I & X & I & X\\
I & Z & I & Z
\end{array}
\right\vert
\begin{array}
[c]{cccc}%
I & I & X & X\\
I & I & Z & Z
\end{array}
.
\]
So the encoding unitary should act as follows:%
\[%
\begin{tabular}
[c]{c|cc|cc}%
\multicolumn{1}{c}{Mem.} & \multicolumn{2}{|c|}{Anc.} &
\multicolumn{2}{|c}{Info.}\\\hline\hline
$I^{\otimes{m}}$ & $Z$ & $I$ & $I$ & $I$\\
$g_{1,1}$ & $I$ & $I$ & $I$ & $I$\\
$g_{1,2}$ & $I$ & $I$ & $I$ & $I$\\
$g_{1,3}$ & $I$ & $I$ & $I$ & $I$\\\hline
$I^{\otimes{m}}$ & $I$ & $Z$ & $I$ & $I$\\
$g_{2,1}$ & $I$ & $I$ & $I$ & $I$\\
$g_{2,2}$ & $I$ & $I$ & $I$ & $I$\\
$g_{2,3}$ & $I$ & $I$ & $I$ & $I$%
\end{tabular}
\ \ \ \ \ \ \rightarrow\ \ \ \ \
\begin{tabular}
[c]{cccc|c}%
\multicolumn{4}{c|}{Phys.} & \multicolumn{1}{|c}{Mem.}\\\hline\hline
$X$ & $X$ & $X$ & $X$ & $g_{1,1}$\\
$X$ & $X$ & $I$ & $I$ & $g_{1,2}$\\
$I$ & $X$ & $I$ & $X$ & $g_{1,3}$\\
$I$ & $I$ & $X$ & $X$ & $I^{\otimes{m}}$\\\hline
$Z$ & $Z$ & $Z$ & $Z$ & $g_{2,1}$\\
$X$ & $X$ & $I$ & $I$ & $g_{2,2}$\\
$I$ & $Z$ & $I$ & $Z$ & $g_{2,3}$\\
$I$ & $I$ & $Z$ & $Z$ & $I^{\otimes{m}}$%
\end{tabular}
.
\]

The commutativity matrix is:%
\[
\Omega=%
\begin{bmatrix}
0 & 0 & 0 & 0 & 0 & 0\\
0 & 0 & 0 & 0 & 0 & 1\\
0 & 0 & 0 & 0 & 1 & 0\\
0 & 0 & 0 & 0 & 0 & 0\\
0 & 0 & 1 & 0 & 0 & 0\\
0 & 1 & 0 & 0 & 0 & 0
\end{bmatrix}
,
\]
with dimension equal to six and rank equal to four. A set of generators with
minimal amount of required memory is:
\begin{align*}
g_{1,1}  &  =ZIII,\ \ \ \ g_{1,2}=IIZI,\ \ \ \ g_{1,3}=IIIX,\\
g_{2,1}  &  =IZII,\ \ \ \ g_{2,2}=IIIZ,\ \ \ \ g_{2,3}=IIXI.
\end{align*}

So the minimal-memory encoder performs the following transformation:%

\[%
\begin{tabular}
[c]{cccc|cc|cc}%
\multicolumn{4}{c}{Mem.} & \multicolumn{2}{|c|}{Anc.} &
\multicolumn{2}{|c}{Info.}\\\hline\hline
$I$ & $I$ & $I$ & $I$ & $Z$ & $I$ & $I$ & $I$\\
$Z$ & $I$ & $I$ & $I$ & $I$ & $I$ & $I$ & $I$\\
$I$ & $I$ & $Z$ & $I$ & $I$ & $I$ & $I$ & $I$\\
$I$ & $I$ & $I$ & $X$ & $I$ & $I$ & $I$ & $I$\\\hline
$I$ & $I$ & $I$ & $I$ & $I$ & $Z$ & $I$ & $I$\\
$I$ & $Z$ & $I$ & $I$ & $I$ & $I$ & $I$ & $I$\\
$I$ & $I$ & $I$ & $Z$ & $I$ & $I$ & $I$ & $I$\\
$I$ & $I$ & $X$ & $I$ & $I$ & $I$ & $I$ & $I$%
\end{tabular}
\ \ \ \ \ \ \rightarrow\ \ \ \ \
\begin{tabular}
[c]{cccc|cccc}%
\multicolumn{4}{c|}{Phys.} & \multicolumn{4}{|c}{Mem.}\\\hline\hline
$X$ & $X$ & $X$ & $X$ & $Z$ & $I$ & $I$ & $I$\\
$X$ & $X$ & $I$ & $I$ & $I$ & $I$ & $Z$ & $I$\\
$I$ & $X$ & $I$ & $X$ & $I$ & $I$ & $I$ & $X$\\
$I$ & $I$ & $X$ & $X$ & $I$ & $I$ & $I$ & $I$\\\hline
$Z$ & $Z$ & $Z$ & $Z$ & $I$ & $Z$ & $I$ & $I$\\
$Z$ & $Z$ & $I$ & $I$ & $I$ & $I$ & $I$ & $Z$\\
$I$ & $Z$ & $I$ & $Z$ & $I$ & $I$ & $X$ & $I$\\
$I$ & $I$ & $Z$ & $Z$ & $I$ & $I$ & $I$ & $I$%
\end{tabular}
.
\]

The memory operators which can be a part of catastrophic cycle are the members
of following set:
\[
C=\{Z_{1}^{e_{1}}Z_{2}^{e_{2}}:e_{1},e_{2}\in\{0,1\}\}.
\]

Based on Theorem~\ref{thm:empty-cat}, since the set $S_{1}$ corresponding the
above transformation is an empty set, assigning any basis of set $C$ to output
memories of the rows of $S_{2}$ will make the encoder non-catastrophic. So any
encoder corresponding the following transformation is non-catastrophic.%

\[%
\begin{tabular}
[c]{cccc|cc|cc}%
\multicolumn{4}{c}{Mem.} & \multicolumn{2}{|c|}{Anc.} &
\multicolumn{2}{|c}{Info.}\\\hline\hline
$I$ & $I$ & $I$ & $I$ & $Z$ & $I$ & $I$ & $I$\\
$Z$ & $I$ & $I$ & $I$ & $I$ & $I$ & $I$ & $I$\\
$I$ & $I$ & $Z$ & $I$ & $I$ & $I$ & $I$ & $I$\\
$I$ & $I$ & $I$ & $X$ & $I$ & $I$ & $I$ & $I$\\\hline
$I$ & $I$ & $I$ & $I$ & $I$ & $Z$ & $I$ & $I$\\
$I$ & $Z$ & $I$ & $I$ & $I$ & $I$ & $I$ & $I$\\
$I$ & $I$ & $I$ & $Z$ & $I$ & $I$ & $I$ & $I$\\
$I$ & $I$ & $X$ & $I$ & $I$ & $I$ & $I$ & $I$\\\hline
$I$ & $I$ & $I$ & $I$ & $I$ & $I$ & $X$ & $I$\\
$I$ & $I$ & $I$ & $I$ & $I$ & $I$ & $I$ & $X$%
\end{tabular}
\ \ \ \ \ \ \rightarrow\ \ \ \ \
\begin{tabular}
[c]{cccc|cccc}%
\multicolumn{4}{c|}{Phys.} & \multicolumn{4}{|c}{Mem.}\\\hline\hline
$X$ & $X$ & $X$ & $X$ & $Z$ & $I$ & $I$ & $I$\\
$X$ & $X$ & $I$ & $I$ & $I$ & $I$ & $Z$ & $I$\\
$I$ & $X$ & $I$ & $X$ & $I$ & $I$ & $I$ & $X$\\
$I$ & $I$ & $X$ & $X$ & $I$ & $I$ & $I$ & $I$\\\hline
$Z$ & $Z$ & $Z$ & $Z$ & $I$ & $Z$ & $I$ & $I$\\
$Z$ & $Z$ & $I$ & $I$ & $I$ & $I$ & $I$ & $Z$\\
$I$ & $Z$ & $I$ & $Z$ & $I$ & $I$ & $X$ & $I$\\
$I$ & $I$ & $Z$ & $Z$ & $I$ & $I$ & $I$ & $I$\\\hline
$I$ & $I$ & $I$ & $I$ & $Z$ & $I$ & $I$ & $I$\\
$I$ & $I$ & $I$ & $I$ & $I$ & $Z$ & $I$ & $I$%
\end{tabular}
.
\]

\subsection{Third Example}

The stabilizer generators for the fourth example in Table I in
\cite{ieee2007forney} are as follows:%
\[%
\begin{array}
[c]{ccccc}%
X & X & X & X & X\\
Z & Z & Z & Z & Z
\end{array}
\left\vert
\begin{array}
[c]{ccccc}%
I & I & X & X & I\\
I & I & Z & Z & I
\end{array}
\right\vert \left.
\begin{array}
[c]{ccccc}%
I & X & X & I & X\\
I & Z & Z & I & Z
\end{array}
\right\vert
\begin{array}
[c]{ccccc}%
I & I & I & X & X\\
I & I & I & Z & Z
\end{array}
.
\]
So the encoding unitary should act as follows:%
\[%
\begin{tabular}
[c]{c|cc|ccc}%
Mem. & \multicolumn{2}{|c|}{Anc.} & \multicolumn{3}{|c}{Info.}\\\hline\hline
$I^{\otimes m}$ & $Z$ & $I$ & $I$ & $I$ & $I$\\
$g_{1,1}$ & $I$ & $I$ & $I$ & $I$ & $I$\\
$g_{1,2}$ & $I$ & $I$ & $I$ & $I$ & $I$\\
$g_{1,3}$ & $I$ & $I$ & $I$ & $I$ & $I$\\\hline
$I^{\otimes m}$ & $I$ & $Z$ & $I$ & $I$ & $I$\\
$g_{2,1}$ & $I$ & $I$ & $I$ & $I$ & $I$\\
$g_{2,2}$ & $I$ & $I$ & $I$ & $I$ & $I$\\
$g_{2,3}$ & $I$ & $I$ & $I$ & $I$ & $I$%
\end{tabular}
\ \ \ \ \ \ \ \rightarrow\ \ \ \ \
\begin{tabular}
[c]{ccccc|c}%
\multicolumn{5}{c}{Phys.} & Mem.\\\hline\hline
$X$ & $X$ & $X$ & $X$ & $X$ & $g_{1,1}$\\
$I$ & $I$ & $X$ & $X$ & $I$ & $g_{1,2}$\\
$I$ & $X$ & $X$ & $I$ & $X$ & $g_{1,3}$\\
$I$ & $I$ & $I$ & $X$ & $X$ & $I^{\otimes m}$\\\hline
$Z$ & $Z$ & $Z$ & $Z$ & $Z$ & $g_{2,1}$\\
$I$ & $I$ & $Z$ & $Z$ & $I$ & $g_{2,2}$\\
$I$ & $Z$ & $Z$ & $I$ & $Z$ & $g_{2,3}$\\
$I$ & $I$ & $I$ & $Z$ & $Z$ & $I^{\otimes m}$%
\end{tabular}
\ .
\]
The commutativity matrix is equal to:%
\[
\Omega=%
\begin{bmatrix}
0 & 0 & 0 & 1 & 0 & 1\\
0 & 0 & 0 & 0 & 1 & 1\\
0 & 0 & 0 & 1 & 1 & 0\\
1 & 0 & 1 & 0 & 0 & 0\\
0 & 1 & 1 & 0 & 0 & 0\\
1 & 1 & 0 & 0 & 0 & 0
\end{bmatrix}
.
\]
The dimension of the matrix is six and the rank is four. A set of generators
with minimal amount of required memory is:
\begin{align*}
g_{1,1}  &  =ZZII,\ \ \ \ g_{1,2}=IIIZ,\ \ \ \ g_{1,3}=IIXI,\\
g_{2,1}  &  =XIZI,\ \ \ \ g_{2,2}=IIZX,\ \ \ \ g_{2,3}=IXIX.
\end{align*}

So the encoding unitary should act as follows:%
\[%
\begin{tabular}
[c]{cccc|cc|ccc}%
\multicolumn{4}{c}{Mem.} & \multicolumn{2}{|c|}{Anc.} &
\multicolumn{3}{|c}{Info.}\\\hline\hline
$I$ & $I$ & $I$ & $I$ & $Z$ & $I$ & $I$ & $I$ & $I$\\
$Z$ & $Z$ & $I$ & $I$ & $I$ & $I$ & $I$ & $I$ & $I$\\
$I$ & $I$ & $I$ & $Z$ & $I$ & $I$ & $I$ & $I$ & $I$\\
$I$ & $I$ & $X$ & $I$ & $I$ & $I$ & $I$ & $I$ & $I$\\\hline
$I$ & $I$ & $I$ & $I$ & $I$ & $Z$ & $I$ & $I$ & $I$\\
$X$ & $I$ & $Z$ & $I$ & $I$ & $I$ & $I$ & $I$ & $I$\\
$I$ & $I$ & $Z$ & $X$ & $I$ & $I$ & $I$ & $I$ & $I$\\
$I$ & $X$ & $I$ & $X$ & $I$ & $I$ & $I$ & $I$ & $I$%
\end{tabular}
\ \ \ \ \ \ \rightarrow\ \ \ \ \
\begin{tabular}
[c]{ccccc|cccc}%
\multicolumn{5}{c|}{Phys.} & \multicolumn{4}{|c}{Mem.}\\\hline\hline
$X$ & $X$ & $X$ & $X$ & $X$ & $Z$ & $Z$ & $I$ & $I$\\
$I$ & $I$ & $X$ & $X$ & $I$ & $I$ & $I$ & $I$ & $Z$\\
$I$ & $X$ & $X$ & $I$ & $X$ & $I$ & $I$ & $X$ & $I$\\
$I$ & $I$ & $I$ & $X$ & $X$ & $I$ & $I$ & $I$ & $I$\\\hline
$Z$ & $Z$ & $Z$ & $Z$ & $Z$ & $X$ & $I$ & $Z$ & $I$\\
$I$ & $I$ & $Z$ & $Z$ & $I$ & $I$ & $I$ & $Z$ & $X$\\
$I$ & $Z$ & $Z$ & $I$ & $Z$ & $I$ & $X$ & $I$ & $X$\\
$I$ & $I$ & $I$ & $Z$ & $Z$ & $I$ & $I$ & $I$ & $I$%
\end{tabular}
.
\]

The memory operators which can be a part of catastrophic cycle are the members
of following set:
\[
C=\{X_{1}^{e_{1}}Z_{1}^{e_{2}}X_{2}^{e_{1}}Z_{2}^{e_{2}}X_{3}^{e_{2}}%
Z_{4}^{e_{2}}:e_{1},e_{2}\in\{0,1\}\}.
\]

Based on Theorem~\ref{thm:empty-cat}, since the set $S_{1}$ corresponding the
above transformation is an empty set, assigning any basis of set $C$ to output
memories of the rows of $S_{2}$ will make the encoder non-catastrophic. So any
encoder implementing the following transformation is non-catastrophic:%

\[%
\begin{tabular}
[c]{cccc|cc|ccc}%
\multicolumn{4}{c}{Mem.} & \multicolumn{2}{|c|}{Anc.} &
\multicolumn{3}{|c}{Info.}\\\hline\hline
$I$ & $I$ & $I$ & $I$ & $Z$ & $I$ & $I$ & $I$ & $I$\\
$Z$ & $Z$ & $I$ & $I$ & $I$ & $I$ & $I$ & $I$ & $I$\\
$I$ & $I$ & $I$ & $Z$ & $I$ & $I$ & $I$ & $I$ & $I$\\
$I$ & $I$ & $X$ & $I$ & $I$ & $I$ & $I$ & $I$ & $I$\\\hline
$I$ & $I$ & $I$ & $I$ & $I$ & $Z$ & $I$ & $I$ & $I$\\
$X$ & $I$ & $Z$ & $I$ & $I$ & $I$ & $I$ & $I$ & $I$\\
$I$ & $I$ & $Z$ & $X$ & $I$ & $I$ & $I$ & $I$ & $I$\\
$I$ & $X$ & $I$ & $X$ & $I$ & $I$ & $I$ & $I$ & $I$\\\hline
$I$ & $I$ & $I$ & $I$ & $I$ & $I$ & $X$ & $I$ & $I$\\
$I$ & $I$ & $I$ & $I$ & $I$ & $I$ & $I$ & $X$ & $I$%
\end{tabular}
\ \ \ \ \ \ \ \rightarrow\ \ \ \ \
\begin{tabular}
[c]{ccccc|cccc}%
\multicolumn{5}{c|}{Phys.} & \multicolumn{4}{|c}{Mem.}\\\hline\hline
$X$ & $X$ & $X$ & $X$ & $X$ & $Z$ & $Z$ & $I$ & $I$\\
$I$ & $I$ & $X$ & $X$ & $I$ & $I$ & $I$ & $I$ & $Z$\\
$I$ & $X$ & $X$ & $I$ & $X$ & $I$ & $I$ & $X$ & $I$\\
$I$ & $I$ & $I$ & $X$ & $X$ & $I$ & $I$ & $I$ & $I$\\\hline
$Z$ & $Z$ & $Z$ & $Z$ & $Z$ & $X$ & $I$ & $Z$ & $I$\\
$I$ & $I$ & $Z$ & $Z$ & $I$ & $I$ & $I$ & $Z$ & $X$\\
$I$ & $Z$ & $Z$ & $I$ & $Z$ & $I$ & $X$ & $I$ & $X$\\
$I$ & $I$ & $I$ & $Z$ & $Z$ & $I$ & $I$ & $I$ & $I$\\\hline
$I$ & $I$ & $I$ & $I$ & $I$ & $X$ & $X$ & $I$ & $I$\\
$I$ & $I$ & $I$ & $I$ & $I$ & $Z$ & $Z$ & $X$ & $Z$%
\end{tabular}
\ .
\]

\subsection{Fourth Example}

The stabilizer generators for the sixth example in Table I in
\cite{ieee2007forney} are:%

\[%
\begin{array}
[c]{ccccc}%
X & X & X & X & X\\
Z & Z & Z & Z & Z
\end{array}
\left\vert
\begin{array}
[c]{ccccc}%
X & I & X & I & X\\
Z & I & Z & I & Z
\end{array}
\right\vert \left.
\begin{array}
[c]{ccccc}%
I & I & I & X & X\\
I & I & I & Z & Z
\end{array}
\right\vert
\begin{array}
[c]{ccccc}%
I & X & X & X & X\\
I & Z & Z & Z & Z
\end{array}
.
\]
So the encoding unitary should act as follows:%
\[%
\begin{tabular}
[c]{c|cc|ccc}%
Mem. & \multicolumn{2}{|c|}{Anc.} & \multicolumn{3}{|c}{Info.}\\\hline\hline
$I^{\otimes m}$ & $Z$ & $I$ & $I$ & $I$ & $I$\\
$g_{1,1}$ & $I$ & $I$ & $I$ & $I$ & $I$\\
$g_{1,2}$ & $I$ & $I$ & $I$ & $I$ & $I$\\
$g_{1,3}$ & $I$ & $I$ & $I$ & $I$ & $I$\\\hline
$I^{\otimes m}$ & $I$ & $Z$ & $I$ & $I$ & $I$\\
$g_{2,1}$ & $I$ & $I$ & $I$ & $I$ & $I$\\
$g_{2,2}$ & $I$ & $I$ & $I$ & $I$ & $I$\\
$g_{2,3}$ & $I$ & $I$ & $I$ & $I$ & $I$%
\end{tabular}
\ \ \ \ \ \ \ \rightarrow\ \ \ \ \
\begin{tabular}
[c]{ccccc|c}%
\multicolumn{5}{c|}{Phys.} & Mem.\\\hline\hline
$X$ & $X$ & $X$ & $X$ & $X$ & $g_{1,1}$\\
$X$ & $I$ & $X$ & $I$ & $X$ & $g_{1,2}$\\
$I$ & $I$ & $I$ & $X$ & $X$ & $g_{1,3}$\\
$I$ & $X$ & $X$ & $X$ & $X$ & $I^{\otimes m}$\\\hline
$Z$ & $Z$ & $Z$ & $Z$ & $Z$ & $g_{2,1}$\\
$Z$ & $I$ & $Z$ & $I$ & $Z$ & $g_{2,2}$\\
$I$ & $I$ & $I$ & $Z$ & $Z$ & $g_{2,3}$\\
$I$ & $Z$ & $Z$ & $Z$ & $Z$ & $I^{\otimes m}$%
\end{tabular}
\ .
\]

The commutativity matrix is:%

\[
\Omega=%
\begin{bmatrix}
0 & 0 & 0 & 1 & 1 & 0\\
0 & 0 & 0 & 1 & 0 & 0\\
0 & 0 & 0 & 0 & 0 & 0\\
1 & 1 & 0 & 0 & 0 & 0\\
1 & 0 & 0 & 0 & 0 & 0\\
0 & 0 & 0 & 0 & 0 & 0
\end{bmatrix}
.
\]
The dimension of $\Omega$ is equal to six and its rank is equal to four. A set
of memory operators which act on minimal amount of required memory is:
\begin{align*}
g_{1,1}  &  =ZZII,\ \ \ \ g_{1,2}=ZIII,\ \ \ \ g_{1,3}=IIZI,\\
g_{2,1}  &  =XIII,\ \ \ \ g_{2,2}=IXII,\ \ \ \ g_{2,3}=IIIZ,
\end{align*}
so the minimal-memory encoder implements the following transformation:
\[%
\begin{tabular}
[c]{cccc|cc|ccc}%
\multicolumn{4}{c}{Mem.} & \multicolumn{2}{|c|}{Anc.} &
\multicolumn{3}{|c}{Info.}\\\hline\hline
$I$ & $I$ & $I$ & $I$ & $Z$ & $I$ & $I$ & $I$ & $I$\\
$Z$ & $Z$ & $I$ & $I$ & $I$ & $I$ & $I$ & $I$ & $I$\\
$Z$ & $I$ & $I$ & $I$ & $I$ & $I$ & $I$ & $I$ & $I$\\
$I$ & $I$ & $Z$ & $I$ & $I$ & $I$ & $I$ & $I$ & $I$\\\hline
$I$ & $I$ & $I$ & $I$ & $I$ & $Z$ & $I$ & $I$ & $I$\\
$X$ & $I$ & $I$ & $I$ & $I$ & $I$ & $I$ & $I$ & $I$\\
$I$ & $X$ & $I$ & $I$ & $I$ & $I$ & $I$ & $I$ & $I$\\
$I$ & $I$ & $I$ & $Z$ & $I$ & $I$ & $I$ & $I$ & $I$%
\end{tabular}
\ \ \ \ \ \ \ \rightarrow\ \ \ \ \
\begin{tabular}
[c]{ccccc|cccc}%
\multicolumn{5}{c|}{Phys.} & \multicolumn{4}{|c}{Mem.}\\\hline\hline
$X$ & $X$ & $X$ & $X$ & $X$ & $Z$ & $Z$ & $I$ & $I$\\
$X$ & $I$ & $X$ & $I$ & $X$ & $Z$ & $I$ & $I$ & $I$\\
$I$ & $I$ & $I$ & $X$ & $X$ & $I$ & $I$ & $Z$ & $I$\\
$I$ & $X$ & $X$ & $X$ & $X$ & $I$ & $I$ & $I$ & $I$\\\hline
$Z$ & $Z$ & $Z$ & $Z$ & $Z$ & $X$ & $I$ & $I$ & $I$\\
$Z$ & $I$ & $Z$ & $I$ & $Z$ & $I$ & $X$ & $I$ & $I$\\
$I$ & $I$ & $I$ & $Z$ & $Z$ & $I$ & $I$ & $I$ & $Z$\\
$I$ & $Z$ & $Z$ & $Z$ & $Z$ & $I$ & $I$ & $I$ & $I$%
\end{tabular}
\ .
\]
The memory operators which can be a part of catastrophic cycle are the members
of following set:
\[
C=\{Z_{3}^{e_{1}}Z_{4}^{e_{2}}:e_{1},e_{2}\in\{0,1\}\}.
\]

Based on Theorem~\ref{thm:empty-cat}, since the set $S_{1}$ corresponding the
above transformation is an empty set, assigning any basis of set $C$ to output
memories of the rows of $S_{2}$ will make the encoder non-catastrophic. So any
encoder implementing the following transformation is non-catastrophic:%

\[%
\begin{tabular}
[c]{cccc|cc|ccc}%
\multicolumn{4}{c}{Mem.} & \multicolumn{2}{|c|}{Anc.} &
\multicolumn{3}{|c}{Info.}\\\hline\hline
$I$ & $I$ & $I$ & $I$ & $Z$ & $I$ & $I$ & $I$ & $I$\\
$Z$ & $Z$ & $I$ & $I$ & $I$ & $I$ & $I$ & $I$ & $I$\\
$Z$ & $I$ & $I$ & $I$ & $I$ & $I$ & $I$ & $I$ & $I$\\
$I$ & $I$ & $Z$ & $I$ & $I$ & $I$ & $I$ & $I$ & $I$\\\hline
$I$ & $I$ & $I$ & $I$ & $I$ & $Z$ & $I$ & $I$ & $I$\\
$X$ & $I$ & $I$ & $I$ & $I$ & $I$ & $I$ & $I$ & $I$\\
$I$ & $X$ & $I$ & $I$ & $I$ & $I$ & $I$ & $I$ & $I$\\
$I$ & $I$ & $I$ & $Z$ & $I$ & $I$ & $I$ & $I$ & $I$\\\hline
$I$ & $I$ & $I$ & $I$ & $I$ & $I$ & $X$ & $I$ & $I$\\
$I$ & $I$ & $I$ & $I$ & $I$ & $I$ & $I$ & $X$ & $I$%
\end{tabular}
\ \ \ \ \ \ \rightarrow\ \ \ \ \
\begin{tabular}
[c]{ccccc|cccc}%
\multicolumn{5}{c|}{Phys.} & \multicolumn{4}{|c}{Mem.}\\\hline\hline
$X$ & $X$ & $X$ & $X$ & $X$ & $Z$ & $Z$ & $I$ & $I$\\
$X$ & $I$ & $X$ & $I$ & $X$ & $Z$ & $I$ & $I$ & $I$\\
$I$ & $I$ & $I$ & $X$ & $X$ & $I$ & $I$ & $Z$ & $I$\\
$I$ & $X$ & $X$ & $X$ & $X$ & $I$ & $I$ & $I$ & $I$\\\hline
$Z$ & $Z$ & $Z$ & $Z$ & $Z$ & $X$ & $I$ & $I$ & $I$\\
$Z$ & $I$ & $Z$ & $I$ & $Z$ & $I$ & $X$ & $I$ & $I$\\
$I$ & $I$ & $I$ & $Z$ & $Z$ & $I$ & $I$ & $I$ & $Z$\\
$I$ & $Z$ & $Z$ & $Z$ & $Z$ & $I$ & $I$ & $I$ & $I$\\\hline
$I$ & $I$ & $I$ & $I$ & $I$ & $I$ & $I$ & $Z$ & $I$\\
$I$ & $I$ & $I$ & $I$ & $I$ & $I$ & $I$ & $I$ & $Z$%
\end{tabular}
.
\]

\subsection{Fifth Example}

The generators for the eighth example in Table I in~\cite{ieee2007forney} are:%
\[%
\begin{array}
[c]{c}%
XXXXXXXX\\
ZZZZZZZZ
\end{array}
\left\vert
\begin{array}
[c]{c}%
IXIXIXIX\\
IZIZIZIZ
\end{array}
\right\vert \left.
\begin{array}
[c]{c}%
IIXXIIXX\\
IIZZIIZZ
\end{array}
\right\vert
\begin{array}
[c]{c}%
IIIIXXXX\\
IIIIZZZZ
\end{array}
.
\]
So the encoder should act as follows:%
\[%
\begin{tabular}
[c]{c|cc|cccccc}%
Mem. & \multicolumn{2}{|c|}{Anc.} & \multicolumn{6}{|c}{Info.}\\\hline\hline
$I^{\otimes m}$ & $Z$ & $I$ & $I$ & $I$ & $I$ & $I$ & $I$ & $I$\\
$g_{1,1}$ & $I$ & $I$ & $I$ & $I$ & $I$ & $I$ & $I$ & $I$\\
$g_{1,2}$ & $I$ & $I$ & $I$ & $I$ & $I$ & $I$ & $I$ & $I$\\
$g_{1,3}$ & $I$ & $I$ & $I$ & $I$ & $I$ & $I$ & $I$ & $I$\\\hline
$I^{\otimes m}$ & $I$ & $Z$ & $I$ & $I$ & $I$ & $I$ & $I$ & $I$\\
$g_{2,1}$ & $I$ & $I$ & $I$ & $I$ & $I$ & $I$ & $I$ & $I$\\
$g_{2,2}$ & $I$ & $I$ & $I$ & $I$ & $I$ & $I$ & $I$ & $I$\\
$g_{2,3}$ & $I$ & $I$ & $I$ & $I$ & $I$ & $I$ & $I$ & $I$%
\end{tabular}
\rightarrow%
\begin{tabular}
[c]{cccccccc|c}%
\multicolumn{8}{c|}{Phys.} & Mem.\\\hline\hline
$X$ & $X$ & $X$ & $X$ & $X$ & $X$ & $X$ & $X$ & $g_{1,1}$\\
$I$ & $X$ & $I$ & $X$ & $I$ & $X$ & $I$ & $X$ & $g_{1,2}$\\
$I$ & $I$ & $X$ & $X$ & $I$ & $I$ & $X$ & $X$ & $g_{1,3}$\\
$I$ & $I$ & $I$ & $I$ & $X$ & $X$ & $X$ & $X$ & $I^{\otimes m}$\\\hline
$Z$ & $Z$ & $Z$ & $Z$ & $Z$ & $Z$ & $Z$ & $Z$ & $g_{2,1}$\\
$I$ & $Z$ & $I$ & $Z$ & $I$ & $Z$ & $I$ & $Z$ & $g_{2,2}$\\
$I$ & $I$ & $Z$ & $Z$ & $I$ & $I$ & $Z$ & $Z$ & $g_{2,3}$\\
$I$ & $I$ & $I$ & $I$ & $Z$ & $Z$ & $Z$ & $Z$ & $I^{\otimes m}$%
\end{tabular}
.
\]
The commutativity matrix is zero matrix, so the minimal number of required
memory is six. A set of memory operators with minimal amount of required
memory is as follows:%

\begin{align*}
g_{1,1} &  =ZIIIII,\ \ \ \ g_{1,2}=IZIIII,\ \ \ \ g_{1,3}=IIZII,\\
g_{2,1} &  =IIIZII,\ \ \ \ g_{2,2}=IIIIZI,\ \ \ \ g_{2,3}=IIIIIZ.
\end{align*}
So the encoder implements the following transformation:%
\[%
\begin{tabular}
[c]{c|cc|cccccc}%
Mem. & \multicolumn{2}{|c|}{Anc.} & \multicolumn{6}{|c}{Info.}\\\hline\hline
$IIIIII$ & $Z$ & $I$ & $I$ & $I$ & $I$ & $I$ & $I$ & $I$\\
$ZIIIII$ & $I$ & $I$ & $I$ & $I$ & $I$ & $I$ & $I$ & $I$\\
$IZIIII$ & $I$ & $I$ & $I$ & $I$ & $I$ & $I$ & $I$ & $I$\\
$IIZIII$ & $I$ & $I$ & $I$ & $I$ & $I$ & $I$ & $I$ & $I$\\\hline
$IIIIII$ & $I$ & $Z$ & $I$ & $I$ & $I$ & $I$ & $I$ & $I$\\
$IIIZII$ & $I$ & $I$ & $I$ & $I$ & $I$ & $I$ & $I$ & $I$\\
$IIIIZI$ & $I$ & $I$ & $I$ & $I$ & $I$ & $I$ & $I$ & $I$\\
$IIIIIZ$ & $I$ & $I$ & $I$ & $I$ & $I$ & $I$ & $I$ & $I$%
\end{tabular}
\rightarrow%
\begin{tabular}
[c]{cccccccc|cccccc}%
\multicolumn{8}{c|}{Phys.} & \multicolumn{6}{|c}{Mem.}\\\hline\hline
$X$ & $X$ & $X$ & $X$ & $X$ & $X$ & $X$ & $X$ & $Z$ & $I$ & $I$ & $I$ & $I$ &
$I$\\
$I$ & $X$ & $I$ & $X$ & $I$ & $X$ & $I$ & $X$ & $I$ & $Z$ & $I$ & $I$ & $I$ &
$I$\\
$I$ & $I$ & $X$ & $X$ & $I$ & $I$ & $X$ & $X$ & $I$ & $I$ & $Z$ & $I$ & $I$ &
$I$\\
$I$ & $I$ & $I$ & $I$ & $X$ & $X$ & $X$ & $X$ & $I$ & $I$ & $I$ & $I$ & $I$ &
$I$\\\hline
$Z$ & $Z$ & $Z$ & $Z$ & $Z$ & $Z$ & $Z$ & $Z$ & $I$ & $I$ & $I$ & $Z$ & $I$ &
$I$\\
$I$ & $Z$ & $I$ & $Z$ & $I$ & $Z$ & $I$ & $Z$ & $I$ & $I$ & $I$ & $I$ & $Z$ &
$I$\\
$I$ & $I$ & $Z$ & $Z$ & $I$ & $I$ & $Z$ & $Z$ & $I$ & $I$ & $I$ & $I$ & $I$ &
$Z$\\
$I$ & $I$ & $I$ & $I$ & $Z$ & $Z$ & $Z$ & $Z$ & $I$ & $I$ & $I$ & $I$ & $I$ &
$I$%
\end{tabular}
.
\]

The memory operators which can be a part of catastrophic cycle are the members
of following set:
\[
C=\{Z_{1}^{e_{1}}Z_{2}^{e_{2}}Z_{3}^{e_{3}}Z_{4}^{e_{4}}Z_{5}^{e_{5}}%
Z_{6}^{e_{6}}:e_{1},e_{2},e_{3},e_{4},e_{5},e_{6}\in\{0,1\}\}.
\]

Based on Theorem~\ref{thm:empty-cat}, since the set $S_{1}$ corresponding the
above transformation is an empty set, assigning any basis of set $C$ to output
memories of the rows of $S_{2}$ will make the encoder non-catastrophic. So any
encoder implementing the following transformation is non-catastrophic:

{\normalsize
\[%
\begin{tabular}
[c]{c|c|cccccc}%
Mem. & \multicolumn{1}{|c|}{Anc.} & \multicolumn{6}{|c}{Info.}\\\hline\hline
$IIIIII$ & $ZI$ & $I$ & $I$ & $I$ & $I$ & $I$ & $I$\\
$ZIIIII$ & $II$ & $I$ & $I$ & $I$ & $I$ & $I$ & $I$\\
$IZIIII$ & $II$ & $I$ & $I$ & $I$ & $I$ & $I$ & $I$\\
$IIZIII$ & $II$ & $I$ & $I$ & $I$ & $I$ & $I$ & $I$\\\hline
$IIIIII$ & $IZ$ & $I$ & $I$ & $I$ & $I$ & $I$ & $I$\\
$IIIZII$ & $II$ & $I$ & $I$ & $I$ & $I$ & $I$ & $I$\\
$IIIIZI$ & $II$ & $I$ & $I$ & $I$ & $I$ & $I$ & $I$\\
$IIIIIZ$ & $II$ & $I$ & $I$ & $I$ & $I$ & $I$ & $I$\\\hline
$IIIIII$ & $II$ & $X$ & $I$ & $I$ & $I$ & $I$ & $I$\\
$IIIIII$ & $II$ & $I$ & $X$ & $I$ & $I$ & $I$ & $I$\\
$IIIIII$ & $II$ & $I$ & $I$ & $X$ & $I$ & $I$ & $I$\\
$IIIIII$ & $II$ & $I$ & $I$ & $I$ & $X$ & $I$ & $I$\\
$IIIIII$ & $II$ & $I$ & $I$ & $I$ & $I$ & $X$ & $I$\\
$IIIIII$ & $II$ & $I$ & $I$ & $I$ & $I$ & $I$ & $X$%
\end{tabular}
\rightarrow%
\begin{tabular}
[c]{cccccccc|cccccc}%
\multicolumn{8}{c|}{Phys.} & \multicolumn{6}{|c}{Mem.}\\\hline\hline
$X$ & $X$ & $X$ & $X$ & $X$ & $X$ & $X$ & $X$ & $Z$ & $I$ & $I$ & $I$ & $I$ &
$I$\\
$I$ & $X$ & $I$ & $X$ & $I$ & $X$ & $I$ & $X$ & $I$ & $Z$ & $I$ & $I$ & $I$ &
$I$\\
$I$ & $I$ & $X$ & $X$ & $I$ & $I$ & $X$ & $X$ & $I$ & $I$ & $Z$ & $I$ & $I$ &
$I$\\
$I$ & $I$ & $I$ & $I$ & $X$ & $X$ & $X$ & $X$ & $I$ & $I$ & $I$ & $I$ & $I$ &
$I$\\\hline
$Z$ & $Z$ & $Z$ & $Z$ & $Z$ & $Z$ & $Z$ & $Z$ & $I$ & $I$ & $I$ & $Z$ & $I$ &
$I$\\
$I$ & $Z$ & $I$ & $Z$ & $I$ & $Z$ & $I$ & $Z$ & $I$ & $I$ & $I$ & $I$ & $Z$ &
$I$\\
$I$ & $I$ & $Z$ & $Z$ & $I$ & $I$ & $Z$ & $Z$ & $I$ & $I$ & $I$ & $I$ & $I$ &
$Z$\\
$I$ & $I$ & $I$ & $I$ & $Z$ & $Z$ & $Z$ & $Z$ & $I$ & $I$ & $I$ & $I$ & $I$ &
$I$\\\hline
$I$ & $I$ & $I$ & $I$ & $I$ & $I$ & $I$ & $I$ & $Z$ & $I$ & $I$ & $I$ & $I$ &
$I$\\
$I$ & $I$ & $I$ & $I$ & $I$ & $I$ & $I$ & $I$ & $I$ & $Z$ & $I$ & $I$ & $I$ &
$I$\\
$I$ & $I$ & $I$ & $I$ & $I$ & $I$ & $I$ & $I$ & $I$ & $I$ & $Z$ & $I$ & $I$ &
$I$\\
$I$ & $I$ & $I$ & $I$ & $I$ & $I$ & $I$ & $I$ & $I$ & $I$ & $I$ & $Z$ & $I$ &
$I$\\
$I$ & $I$ & $I$ & $I$ & $I$ & $I$ & $I$ & $I$ & $I$ & $I$ & $I$ & $I$ & $Z$ &
$I$\\
$I$ & $I$ & $I$ & $I$ & $I$ & $I$ & $I$ & $I$ & $I$ & $I$ & $I$ & $I$ & $I$ &
$Z$%
\end{tabular}
\ .
\]
}

\subsection{Sixth Example}

The example in the third row of Figure 1 in~\cite{GR07} has the following
generators:
\[%
\begin{array}
[c]{cccc}%
X & X & X & X\\
Z & Z & Z & Z
\end{array}
\left\vert
\begin{array}
[c]{cccc}%
I & I & X & X\\
I & I & Z & Z
\end{array}
\right\vert
\begin{array}
[c]{cccc}%
I & X & I & X\\
I & Z & I & Z
\end{array}
\left\vert
\begin{array}
[c]{cccc}%
I & I & X & X\\
I & I & Z & Z
\end{array}
\right\vert
\begin{array}
[c]{cccc}%
X & X & X & X\\
Z & Z & Z & Z
\end{array}
.
\]
So the encoder should act as follows:%
\[%
\begin{tabular}
[c]{c|cc|cc}%
Mem. & \multicolumn{2}{|c|}{Anc.} & \multicolumn{2}{|c}{Info.}\\\hline\hline
$I^{\otimes m}$ & $Z$ & $I$ & $I$ & $I$\\
$g_{1,1}$ & $I$ & $I$ & $I$ & $I$\\
$g_{1,2}$ & $I$ & $I$ & $I$ & $I$\\
$g_{1,3}$ & $I$ & $I$ & $I$ & $I$\\
$g_{1,4}$ & $I$ & $I$ & $I$ & $I$\\\hline
$I^{\otimes m}$ & $I$ & $Z$ & $I$ & $I$\\
$g_{2,1}$ & $I$ & $I$ & $I$ & $I$\\
$g_{2,2}$ & $I$ & $I$ & $I$ & $I$\\
$g_{2,3}$ & $I$ & $I$ & $I$ & $I$\\
$g_{2,4}$ & $I$ & $I$ & $I$ & $I$%
\end{tabular}
\rightarrow%
\begin{tabular}
[c]{cccc|c}%
\multicolumn{4}{c|}{Phys.} & Mem.\\\hline\hline
$X$ & $X$ & $X$ & $X$ & $g_{1,1}$\\
$I$ & $I$ & $X$ & $X$ & $g_{1,2}$\\
$I$ & $X$ & $I$ & $X$ & $g_{1,3}$\\
$I$ & $I$ & $X$ & $X$ & $g_{1,4}$\\
$X$ & $X$ & $X$ & $X$ & $I^{\otimes m}$\\\hline
$Z$ & $Z$ & $Z$ & $Z$ & $g_{2,1}$\\
$I$ & $I$ & $Z$ & $Z$ & $g_{2,2}$\\
$I$ & $Z$ & $I$ & $Z$ & $g_{2,3}$\\
$I$ & $I$ & $Z$ & $Z$ & $g_{2,4}$\\
$Z$ & $Z$ & $Z$ & $Z$ & $I^{\otimes m}$%
\end{tabular}
.
\]

The commutativity matrix is:%
\[
\Omega=%
\begin{bmatrix}
0 & 0 & 0 & 0 & 0 & 0 & 0 & 0\\
0 & 0 & 0 & 0 & 0 & 0 & 1 & 0\\
0 & 0 & 0 & 0 & 0 & 1 & 0 & 0\\
0 & 0 & 0 & 0 & 0 & 0 & 0 & 0\\
0 & 0 & 0 & 0 & 0 & 0 & 0 & 0\\
0 & 0 & 1 & 0 & 0 & 0 & 0 & 0\\
0 & 1 & 0 & 0 & 0 & 0 & 0 & 0\\
0 & 0 & 0 & 0 & 0 & 0 & 0 & 0
\end{bmatrix}
,
\]
with dimension equal to eight and rank equal to six. So the minimal amount of
required memory is five. A set of memory operators with minimal number of
required memory is:
\begin{align*}
g_{1,1} &  =ZIIIII,\ \ \ g_{1,2}=IIXIII,\ \ \ g_{1,3}=IIIZII,\ \ \ g_{1,4}%
=IIIIZI,\\
g_{2,1} &  =IZIIII,\ \ \ g_{2,2}=IIIXII,\ \ \ g_{2,3}=IIZIII,\ \ \ g_{2,4}%
=IIIIIZ.
\end{align*}
Thus, the encoder acts as follows:%
\[%
\begin{tabular}
[c]{cccccc|cc|cc}%
\multicolumn{6}{c}{Mem.} & \multicolumn{2}{|c|}{Anc.} &
\multicolumn{2}{|c}{Info.}\\\hline\hline
$I$ & $I$ & $I$ & $I$ & $I$ & $I$ & $Z$ & $I$ & $I$ & $I$\\
$Z$ & $I$ & $I$ & $I$ & $I$ & $I$ & $I$ & $I$ & $I$ & $I$\\
$I$ & $I$ & $X$ & $I$ & $I$ & $I$ & $I$ & $I$ & $I$ & $I$\\
$I$ & $I$ & $I$ & $Z$ & $I$ & $I$ & $I$ & $I$ & $I$ & $I$\\
$I$ & $I$ & $I$ & $I$ & $Z$ & $I$ & $I$ & $I$ & $I$ & $I$\\\hline
$I$ & $I$ & $I$ & $I$ & $I$ & $I$ & $I$ & $Z$ & $I$ & $I$\\
$I$ & $Z$ & $I$ & $I$ & $I$ & $I$ & $I$ & $I$ & $I$ & $I$\\
$I$ & $I$ & $I$ & $X$ & $I$ & $I$ & $I$ & $I$ & $I$ & $I$\\
$I$ & $I$ & $Z$ & $I$ & $I$ & $I$ & $I$ & $I$ & $I$ & $I$\\
$I$ & $I$ & $I$ & $I$ & $I$ & $Z$ & $I$ & $I$ & $I$ & $I$%
\end{tabular}
\rightarrow%
\begin{tabular}
[c]{cccc|cccccc}%
\multicolumn{4}{c|}{Phys.} & \multicolumn{6}{|c}{Mem.}\\\hline\hline
$X$ & $X$ & $X$ & $X$ & $Z$ & $I$ & $I$ & $I$ & $I$ & $I$\\
$I$ & $I$ & $X$ & $X$ & $I$ & $I$ & $X$ & $I$ & $I$ & $I$\\
$I$ & $X$ & $I$ & $X$ & $I$ & $I$ & $I$ & $Z$ & $I$ & $I$\\
$I$ & $I$ & $X$ & $X$ & $I$ & $I$ & $I$ & $I$ & $Z$ & $I$\\
$X$ & $X$ & $X$ & $X$ & $I$ & $I$ & $I$ & $I$ & $I$ & $I$\\\hline
$Z$ & $Z$ & $Z$ & $Z$ & $I$ & $Z$ & $I$ & $I$ & $I$ & $I$\\
$I$ & $I$ & $Z$ & $Z$ & $I$ & $I$ & $I$ & $X$ & $I$ & $I$\\
$I$ & $Z$ & $I$ & $Z$ & $I$ & $I$ & $Z$ & $I$ & $I$ & $I$\\
$I$ & $I$ & $Z$ & $Z$ & $I$ & $I$ & $I$ & $I$ & $I$ & $Z$\\
$Z$ & $Z$ & $Z$ & $Z$ & $I$ & $I$ & $I$ & $I$ & $I$ & $I$%
\end{tabular}
.
\]
The memory operators which can be a part of catastrophic cycle are the members
of following set:
\[
C=\{Z_{1}^{e_{1}}Z_{2}^{e_{2}}Z_{5}^{e_{3}}Z_{6}^{e_{4}}:e_{1},e_{2}%
,e_{3},e_{4}\in\{0,1\}\}.
\]

By inspecting the above transformation, the set $S_{1}$ is follows:%

\[%
\begin{tabular}
[c]{cccccc|cc|cc}%
\multicolumn{6}{c}{Mem.} & \multicolumn{2}{|c|}{Anc.} &
\multicolumn{2}{|c}{Info.}\\\hline\hline
$I$ & $I$ & $I$ & $I$ & $Z$ & $I$ & $Z$ & $I$ & $I$ & $I$\\
$I$ & $I$ & $I$ & $I$ & $I$ & $Z$ & $I$ & $Z$ & $I$ & $I$%
\end{tabular}
\rightarrow%
\begin{tabular}
[c]{cccc|cccccc}%
\multicolumn{4}{c|}{Phys.} & \multicolumn{6}{|c}{Mem.}\\\hline\hline
$I$ & $I$ & $I$ & $I$ & $Z$ & $I$ & $I$ & $I$ & $I$ & $I$\\
$I$ & $I$ & $I$ & $I$ & $I$ & $Z$ & $I$ & $I$ & $I$ & $I$%
\end{tabular}
.
\]
We add two rows (the rows after the line) to the rows of $S_{1}$ as follows:
\[%
\begin{tabular}
[c]{cccccc|cc|cc}%
\multicolumn{6}{c}{Mem.} & \multicolumn{2}{|c|}{Anc.} &
\multicolumn{2}{|c}{Info.}\\\hline\hline
$I$ & $I$ & $I$ & $I$ & $Z$ & $I$ & $Z$ & $I$ & $I$ & $I$\\
$I$ & $I$ & $I$ & $I$ & $I$ & $Z$ & $I$ & $Z$ & $I$ & $I$\\\hline
$I$ & $I$ & $I$ & $I$ & $I$ & $I$ & $I$ & $I$ & $X$ & $I$\\
$I$ & $I$ & $I$ & $I$ & $I$ & $I$ & $I$ & $I$ & $I$ & $X$%
\end{tabular}
\rightarrow%
\begin{tabular}
[c]{cccc|cccccc}%
\multicolumn{4}{c|}{Phys.} & \multicolumn{6}{|c}{Mem.}\\\hline\hline
$I$ & $I$ & $I$ & $I$ & $Z$ & $I$ & $I$ & $I$ & $I$ & $I$\\
$I$ & $I$ & $I$ & $I$ & $I$ & $Z$ & $I$ & $I$ & $I$ & $I$\\\hline
$I$ & $I$ & $I$ & $I$ & $I$ & $I$ & $I$ & $I$ & $Z$ & $I$\\
$I$ & $I$ & $I$ & $I$ & $I$ & $I$ & $I$ & $I$ & $I$ & $Z$%
\end{tabular}
.
\]
All combinations of the entries of $S_{1}\cup S_{2}$ are as follows:
\begin{equation}%
\begin{tabular}
[c]{cccccc|cc|cc}%
\multicolumn{6}{c}{Mem.} & \multicolumn{2}{|c|}{Anc.} &
\multicolumn{2}{|c}{Info.}\\\hline\hline
$I$ & $I$ & $I$ & $I$ & $Z$ & $I$ & $Z$ & $I$ & $I$ & $I$\\
$I$ & $I$ & $I$ & $I$ & $I$ & $Z$ & $I$ & $Z$ & $I$ & $I$\\
$I$ & $I$ & $I$ & $I$ & $I$ & $I$ & $I$ & $I$ & $X$ & $I$\\
$I$ & $I$ & $I$ & $I$ & $I$ & $I$ & $I$ & $I$ & $I$ & $X$\\
$I$ & $I$ & $I$ & $I$ & $Z$ & $Z$ & $Z$ & $Z$ & $I$ & $I$\\
$I$ & $I$ & $I$ & $I$ & $Z$ & $I$ & $Z$ & $I$ & $X$ & $I$\\
$I$ & $I$ & $I$ & $I$ & $Z$ & $I$ & $Z$ & $I$ & $I$ & $X$\\
$I$ & $I$ & $I$ & $I$ & $I$ & $Z$ & $I$ & $Z$ & $X$ & $I$\\
$I$ & $I$ & $I$ & $I$ & $I$ & $Z$ & $I$ & $Z$ & $I$ & $X$\\
$I$ & $I$ & $I$ & $I$ & $I$ & $I$ & $I$ & $I$ & $X$ & $X$\\
$I$ & $I$ & $I$ & $I$ & $Z$ & $Z$ & $Z$ & $Z$ & $X$ & $I$\\
$I$ & $I$ & $I$ & $I$ & $Z$ & $Z$ & $Z$ & $Z$ & $I$ & $X$\\
$I$ & $I$ & $I$ & $I$ & $Z$ & $I$ & $Z$ & $I$ & $X$ & $X$\\
$I$ & $I$ & $I$ & $I$ & $I$ & $Z$ & $I$ & $Z$ & $X$ & $X$\\
$I$ & $I$ & $I$ & $I$ & $Z$ & $Z$ & $Z$ & $Z$ & $X$ & $X$%
\end{tabular}
\rightarrow%
\begin{tabular}
[c]{cccc|cccccc}%
\multicolumn{4}{c|}{Phys.} & \multicolumn{6}{|c}{Mem.}\\\hline\hline
$I$ & $I$ & $I$ & $I$ & $Z$ & $I$ & $I$ & $I$ & $I$ & $I$\\
$I$ & $I$ & $I$ & $I$ & $I$ & $Z$ & $I$ & $I$ & $I$ & $I$\\
$I$ & $I$ & $I$ & $I$ & $I$ & $I$ & $I$ & $I$ & $Z$ & $I$\\
$I$ & $I$ & $I$ & $I$ & $I$ & $I$ & $I$ & $I$ & $I$ & $Z$\\
$I$ & $I$ & $I$ & $I$ & $Z$ & $Z$ & $I$ & $I$ & $I$ & $I$\\
$I$ & $I$ & $I$ & $I$ & $Z$ & $I$ & $I$ & $I$ & $Z$ & $I$\\
$I$ & $I$ & $I$ & $I$ & $Z$ & $I$ & $I$ & $I$ & $I$ & $Z$\\
$I$ & $I$ & $I$ & $I$ & $I$ & $Z$ & $I$ & $I$ & $Z$ & $I$\\
$I$ & $I$ & $I$ & $I$ & $I$ & $Z$ & $I$ & $I$ & $I$ & $Z$\\
$I$ & $I$ & $I$ & $I$ & $I$ & $I$ & $I$ & $I$ & $Z$ & $Z$\\
$I$ & $I$ & $I$ & $I$ & $Z$ & $Z$ & $I$ & $I$ & $Z$ & $I$\\
$I$ & $I$ & $I$ & $I$ & $Z$ & $Z$ & $I$ & $I$ & $I$ & $Z$\\
$I$ & $I$ & $I$ & $I$ & $Z$ & $I$ & $I$ & $I$ & $Z$ & $Z$\\
$I$ & $I$ & $I$ & $I$ & $I$ & $Z$ & $I$ & $I$ & $Z$ & $Z$\\
$I$ & $I$ & $I$ & $I$ & $Z$ & $Z$ & $I$ & $I$ & $Z$ & $Z$%
\end{tabular}
.\label{eq:non-cat-added-row-comb2}%
\end{equation}
By inspecting~(\ref{eq:non-cat-added-row-comb2}) it is obvious that a
catastrophic cycle does not happen. So the encoder which performs the
following transformation is non-catastrophic:
\[%
\begin{tabular}
[c]{cccccc|cc|cc}%
\multicolumn{6}{c}{Mem.} & \multicolumn{2}{|c|}{Anc.} &
\multicolumn{2}{|c}{Info.}\\\hline\hline
$I$ & $I$ & $I$ & $I$ & $I$ & $I$ & $Z$ & $I$ & $I$ & $I$\\
$Z$ & $I$ & $I$ & $I$ & $I$ & $I$ & $I$ & $I$ & $I$ & $I$\\
$I$ & $I$ & $X$ & $I$ & $I$ & $I$ & $I$ & $I$ & $I$ & $I$\\
$I$ & $I$ & $I$ & $Z$ & $I$ & $I$ & $I$ & $I$ & $I$ & $I$\\
$I$ & $I$ & $I$ & $I$ & $Z$ & $I$ & $I$ & $I$ & $I$ & $I$\\\hline
$I$ & $I$ & $I$ & $I$ & $I$ & $I$ & $I$ & $Z$ & $I$ & $I$\\
$I$ & $Z$ & $I$ & $I$ & $I$ & $I$ & $I$ & $I$ & $I$ & $I$\\
$I$ & $I$ & $I$ & $X$ & $I$ & $I$ & $I$ & $I$ & $I$ & $I$\\
$I$ & $I$ & $Z$ & $I$ & $I$ & $I$ & $I$ & $I$ & $I$ & $I$\\
$I$ & $I$ & $I$ & $I$ & $I$ & $Z$ & $I$ & $I$ & $I$ & $I$\\\hline
$I$ & $I$ & $I$ & $I$ & $I$ & $I$ & $I$ & $I$ & $X$ & $I$\\
$I$ & $I$ & $I$ & $I$ & $I$ & $I$ & $I$ & $I$ & $I$ & $X$%
\end{tabular}
\rightarrow%
\begin{tabular}
[c]{cccc|cccccc}%
\multicolumn{4}{c|}{Phys.} & \multicolumn{6}{|c}{Mem.}\\\hline\hline
$X$ & $X$ & $X$ & $X$ & $Z$ & $I$ & $I$ & $I$ & $I$ & $I$\\
$I$ & $I$ & $X$ & $X$ & $I$ & $I$ & $X$ & $I$ & $I$ & $I$\\
$I$ & $X$ & $I$ & $X$ & $I$ & $I$ & $I$ & $Z$ & $I$ & $I$\\
$I$ & $I$ & $X$ & $X$ & $I$ & $I$ & $I$ & $I$ & $Z$ & $I$\\
$X$ & $X$ & $X$ & $X$ & $I$ & $I$ & $I$ & $I$ & $I$ & $I$\\\hline
$Z$ & $Z$ & $Z$ & $Z$ & $I$ & $Z$ & $I$ & $I$ & $I$ & $I$\\
$I$ & $I$ & $Z$ & $Z$ & $I$ & $I$ & $I$ & $X$ & $I$ & $I$\\
$I$ & $Z$ & $I$ & $Z$ & $I$ & $I$ & $Z$ & $I$ & $I$ & $I$\\
$I$ & $I$ & $Z$ & $Z$ & $I$ & $I$ & $I$ & $I$ & $I$ & $Z$\\
$Z$ & $Z$ & $Z$ & $Z$ & $I$ & $I$ & $I$ & $I$ & $I$ & $I$\\\hline
$I$ & $I$ & $I$ & $I$ & $I$ & $I$ & $I$ & $I$ & $Z$ & $I$\\
$I$ & $I$ & $I$ & $I$ & $I$ & $I$ & $I$ & $I$ & $I$ & $Z$%
\end{tabular}
.
\]

\end{document}